%% file: disturbanceConstrained-extendedversion-v4-arxiv7Nov2011.tex
\documentclass[letterpaper,onecolumn]{IEEEtranBB}   

\include{defns}

\usepackage[T1]{fontenc}
\usepackage{graphicx}

\graphicspath{{figs/}}   

\usepackage[cmex10]{amsmath}
\usepackage{amssymb, amsthm, gensymb}  
\usepackage[only,llbracket,rrbracket]{stmaryrd}  

\usepackage{float}
\usepackage{flafter}  
\usepackage{accents}

\input{includesvgNoBuild}  

\usepackage[tight,small,bf]{subfigure}   

\usepackage{microtype}

\usepackage{caption}
\DeclareCaptionLabelSeparator{periodquad}{. \quad}
\usepackage{anysize}
	\marginsize{4cm}{4cm}{2cm}{2cm}  
	
\usepackage[colorlinks=true,linkcolor=black,anchorcolor=black,citecolor=black,filecolor=black,menucolor=black,urlcolor=black]{hyperref}  
 	\hypersetup{
		pdfauthor = {Bernd Bandemer and Abbas El Gamal},
		pdftitle = {Communication with Disturbance Constraints},
		pdfsubject = {Information theory},
		pdfkeywords = {},
		pdfcreator = {LaTeX}
	}
	
\usepackage{cite}  

\newcommand{\T} {^\mathrm{T}}
\newcommand{\Ex}{ \mathbb{E}  }
\renewcommand{\eps}{\varepsilon}
\newcommand{\abs}[1]{\lvert #1 \rvert}

\renewcommand{\det}[1]{\lvert #1 \rvert}

\providecommand{\abs}[1]{\lvert#1\rvert}

\providecommand{\card}[1]{\lvert#1\rvert}

\newcommand{\anneq}[1]{\overset{\text{(#1)}}{=}}  
\newcommand{\annleq}[1]{\overset{\text{(#1)}}{\leq}}  
\newcommand{\anngeq}[1]{\overset{\text{(#1)}}{\geq}}  

\newcommand{\Typ}{\mathcal T_\varepsilon^{(n)}}
\newcommand{\Typprime}{\mathcal T_{\varepsilon'}^{(n)}}
\newcommand{\Er}{\mathcal E}
\newcommand{\cond}{\,|\,}

\newcommand{\eeq}{\mathop{\phantom{\neq}}}
\newcommand{\Es}{\mathcal E_\text{eq}}

\newcommand{\Lc}{\mathcal{L}}

\newcommand{\Rdk}[1]{R_{\text{d},#1}}
\newcommand{\Rd}{R_{\text{d}}}
\newcommand{\Rl}{R_{\text{leak}}}

\newcommand{\natSet}[1]{[1\hspace{-0.25em}:\hspace{-0.25em}#1]}  
\newcommand{\refines}{\preccurlyeq}  
\newcommand{\modn}[1]{\llbracket #1 \rrbracket}  

\newfloat{tabelle}{t}{tabelle}
\floatname{tabelle}{Table}

\theoremstyle{definition}  
\newtheorem{thm}{Theorem}
    
\newtheorem{lemma}{Lemma}
\newtheorem{corollary}{Corollary}

\theoremstyle{remark}
\newtheorem{remark}{Remark}
\newtheorem{example}{Example}

\title{Communication  with \\ Disturbance Constraints}

\author{Bernd Bandemer and Abbas El Gamal \\
Information Systems Laboratory, Stanford University,  \\
350 Serra Mall, Stanford, CA 94305, USA \\
Email: bandemer@stanford.edu, abbas@ee.stanford.edu
	\thanks{\hrule \vspace{2mm} \noindent This work is partially supported by DARPA ITMANET. Bernd Bandemer is supported by an Eric and Illeana Benhamou Stanford Graduate Fellowship.}
}

\begin{document}
\maketitle

\begin{abstract} 

Motivated by the broadcast view of the interference channel, the new problem of communication with disturbance constraints is formulated. The rate--disturbance region is established for the single constraint case and the optimal encoding scheme turns out to be the same as the Han--Kobayashi scheme for the two user-pair interference channel. This result is extended to the Gaussian vector (MIMO) case. For the case of communication with two disturbance constraints, inner and outer bounds on the rate--disturbance region for a deterministic model are established. The inner bound is achieved by an encoding scheme that involves rate splitting, Marton coding, and superposition coding, and is shown to be optimal in several nontrivial cases. This encoding scheme can be readily applied to discrete memoryless interference channels and motivates a natural extension of the Han--Kobayashi scheme to more than two user pairs.

\end{abstract}

\section{Introduction}

Alice wishes to communicate a message to Bob while causing the least disturbance to nearby Dick, Diane, and Diego, who are not interested in the communication from Alice. Assume a discrete memoryless broadcast channel $p(y,z_1,\dots,z_K|x)$ between Alice $X$, Bob $Y$, and their preoccupied friends $Z_1,\dots,Z_K$ as depicted in Figure~\ref{fig:disturbanceChannelK_oneBox}. We measure the disturbance at side receiver $Z_j$ by the amount of {\em undesired} information rate $(1/n)I(X^n;Z_j^n)$ originating from the sender $X$, and require this rate not to exceed $\Rdk j$ in the limit.  The problem is to determine the optimal trade-off between the message communication rate $R$ and the disturbance rates $\Rdk j$. 

\begin{figure}[h]
	\centering
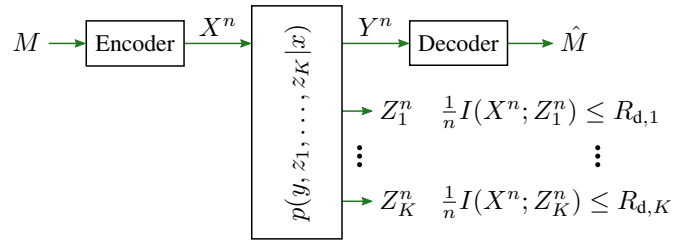%

	\caption{Communication system with disturbance constraints.}
	\label{fig:disturbanceChannelK_oneBox}
\end{figure}
This communication with disturbance constraints problem is motivated by the broadcast  side of the interference channel in which each sender wishes to communicate a message only to one of the receivers while causing the least disturbance to the other receivers. However, in this paper, which is an extended version of~\cite{Bandemer:2011:Disturbance}, we focus on studying the problem of communication with disturbance constraints itself. The application of the coding scheme developed in this paper to deterministic interference channels with more than two user pairs is discussed in~\cite{Bandemer:2011:3DicMarton}. 

For a single disturbance constraint, we show that the optimal encoding scheme is rate splitting and superposition coding, which is the same as the Han--Kobayashi scheme for the two user-pair interference channel~\cite{HanKobayashi81,ChongMotani08}. This motivates us to study communication with more than one disturbance constraint with the hope of finding good coding schemes for interference channels with more than two user pairs. To this end, we establish inner and outer bounds on the rate--disturbance region for the deterministic channel model with two disturbance constraints that are tight in some nontrivial special cases. In the following section we provide needed definitions and present an extended summary of our results. The proofs are presented in subsequent sections, with some parts deferred to the Appendix.


\section{Definitions and main results} \label{sec:defMainresults}
Consider the discrete memoryless communication system with $K$ disturbance constraints (henceforth referred to as DMC-$K$-DC) depicted  in  Figure~\ref{fig:disturbanceChannelK_oneBox}. 
The channel consists of $K+2$ finite alphabets $\Xc$, $\Yc$, $\Zc_j$, $j\in \natSet K$, and a collection of conditional pmfs $p(y,z_1,\dots,z_K| x)$. 
A $(2^{nR},n)$ code for the DMC-$K$-DC consists of the message set $\natSet{2^{nR}}$, an encoding function $x^n: \natSet{2^{nR}} \to \Xc^n$, and a decoding function $\hat m: \Yc^n \to \natSet{2^{nR}}$. We assume that the message $M$ is uniformly distributed over $\natSet{2^{nR}}$. A rate--disturbance tuple $(R,\Rdk 1,\dots,\Rdk K)\in \Real_+^{K+1}$ is achievable for the DMC-$K$-DC if there exists a sequence of $(2^{nR},n)$ codes such that
\begin{align*}
	\lim_{n\to\infty} \P(\hat M \neq M) &= 0, \notag \\
	\limsup_{n\to\infty} \ (1/n) I(X^n;Z_j^n) &\leq \Rdk j, \quad j \in \natSet K.
\end{align*}
The \emph{rate--disturbance region} $\Rr$ of the DMC-$K$-DC is the closure of the set of all achievable tuples $(R,\Rdk 1,\dots,\Rdk K)$. 

\begin{remark} 
	Like the message rate $R$, the disturbance rates $\Rdk j$, for $j\in \natSet K$, are measured in units of bits per channel use. (We use logarithms of base $2$ throughout.)
\end{remark}
\begin{remark}
	The measure of disturbance $(1/n) I(X^n;Z_j^n)$ can be expanded as $(1/n) H(Z_j^n) - (1/n) H(Z_j^n\cond X^n)$. The first term is the entropy rate of the received signal $Z_j$ and is caused by both the transmission itself and by noise inherent to the channel. Subtracting the second term separates out the noise part. (For channels with additive white noise, e.g., the Gaussian case, the second term is exactly the differential entropy of each noise sample.)
\end{remark}
\begin{remark}
	Our results remain essentially true if disturbance is measured by $(1/n) H(Z_j^n)$ instead. If the channel is deterministic, the two measures coincide.
\end{remark}
\begin{remark} 
The disturbance constraint $(1/n) I(X^n;Z_j^n) \leq \Rdk j$ is reminiscent of the information leakage rate constraint for the wiretap channel~\cite{Wyner1975, Csiszar1978}, that is, $(1/n) I(M; Z_j^n) \leq \Rl$. Replacing $M$ with $X^n$, however, dramatically changes the problem and the optimal coding scheme. In the wiretap channel, the key component of the optimal encoding scheme is  randomized encoding, which helps control the leakage rate $(1/n) I(M;Z_j^n)$. Such randomization reduces the achievable transmission rate for a given disturbance constraint, hence is not desirable in our setting.
\end{remark}
The rate--disturbance region is not known in general.  In this paper we establish the following results.

\subsection{Rate--disturbance region for a single disturbance constraint}
Consider the case with a single disturbance constraint, i.e., $K=1$, and relabel $Z_1$ as $Z$ and $\Rdk 1$ as $\Rd$. We  fully characterize the rate--disturbance region for this case. 

\begin{thm} \label{thm:dcCapa}
	The rate--disturbance region $\Rr$ of the DMC-$1$-DC is the set of rate pairs $(R,\Rd)$ such that
	\begin{align*}
		 R  &\leq I(X;Y), \\
		 \Rd &\geq I(X;Z \cond U), \\
		 R -\Rd &\leq I(X;Y\cond U) - I(X;Z\cond U),
	\end{align*}
	for some pmf $p(u,x)$ with $\abs{\Uc} \leq \abs{\Xc} + 1$. 
\end{thm}
Let $\Rr(U,X)$ be the rate region defined by the rate constraints in the theorem for a fixed joint  pmf $(U,X) \sim p(u,x)$. This rate region is illustrated in Figure~\ref{fig:dcCapa_constituentRegion}.  The rate--disturbance region is simply the union of these regions over all $p(u,x)$ and is convex without the need for a time-sharing random variable.

\begin{figure}[htb]
	\centering
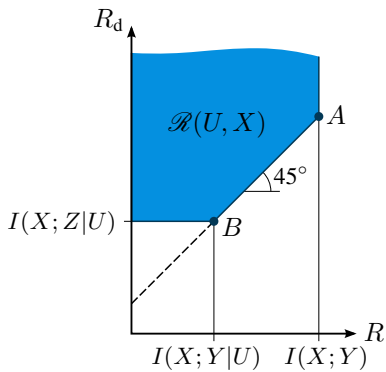%

	\caption{Example of $\Rr(U,X)$, the constituent region of $\Rr$.}
	\label{fig:dcCapa_constituentRegion}
\end{figure}

The proof of Theorem~\ref{thm:dcCapa} is given in Subsections~\ref{sec:achievability_dcCapa} and~\ref{sec:converse_dcCapa}. Achievability is established using rate splitting and superposition coding. Receiver $Y$ decodes the satellite codeword while receiver $Z$ distinguishes only the cloud center. Note that this encoding scheme is identical to the Han--Kobayashi scheme for the two user-pair interference channel~\cite{HanKobayashi81,ChongMotani08}. 

We now consider three interesting special cases.

\medskip
\subsubsection{Deterministic channel} 

Assume that $Y$ and $Z$ are deterministic functions of $X$. We show that the rate--disturbance region in Theorem~\ref{thm:dcCapa} reduces to the following.

\begin{corollary} \label{thm:dcCapa_Det}
	The rate--disturbance region for the deterministic channel with one disturbance constraint is the set of rate pairs $(R,\Rd)$ such that 
	\begin{align*}
		R & \leq H(Y), \\
		R-\Rd & \leq H(Y\cond Z),
	\end{align*}
	for some pmf $p(x)$.
\end{corollary}

Clearly, this region is convex. Alternatively, the region can be written as the set of rate pairs $(R,\Rd)$ such that
\begin{align*}
	R & \leq H(Y\cond Q), \\
	\Rd & \geq I(Y;Z\cond Q),
\end{align*}
for some joint pmf $p(q,x)$ with $\abs{\Qc} \leq 2$. Corollary~\ref{thm:dcCapa_Det} and the alternative description of the region are established by substituting $U=Z$ in the region of Theorem~\ref{thm:dcCapa} and simplifying the resulting region as detailed in Subsection~\ref{sec:proof_dcCapa_Det}. 

\begin{remark} \label{rem:dc_ic_correspondence}
Consider the injective deterministic interference channel with two user pairs depicted in Figure~\ref{fig:injectiveDet2ic_connection2}. Here, $g_{ij}$ is a function that models the link from transmitter $i$ to receiver $j$, for $i,j \in \{1,2\}$. The combining functions $f_j$ are assumed to be injective in each argument. This setting is a special case of the channel investigated in~\cite{ElGamalCosta82}. This can be seen by merging $g_{11}$ and $f_1$ of Figure~\ref{fig:injectiveDet2ic_connection2} into a function $f'_1$ that maps $(X_1,Z_2)$ to $Y_1$. Likewise, define the function $f'_2$ as the merger of $g_{22}$ and $f_2$. The modified combining functions $f'_1$ and $f'_2$ are injective in $Z_2$ and $Z_1$, respectively, and therefore satisfy the assumptions in~\cite{ElGamalCosta82}. It follows that the Han--Kobayashi scheme where the transmitters use superposition codebooks generated according to $p(z_1)p(x_1|z_1)$ and $p(z_2)p(x_2|z_2)$ achieves the capacity region of the channel in Figure~\ref{fig:injectiveDet2ic_connection2}.

On the other hand, Corollary~\ref{thm:dcCapa_Det} shows that the \emph{same} encoding scheme achieves the disturbance-constrained capacity for the channels $X_1 \to (Y'_1,Z_1)$ and $X_2 \to (Y'_2,Z_2)$, shown as dashed boxes in Figure~\ref{fig:injectiveDet2ic_connection2}. Here, $Y'_1$ and $Y'_2$ are the desired receivers, and $Z_1$ and $Z_2$ are the side receivers associated with disturbance constraints. Note that decodability of the desired messages at receivers $Y_1$ and $Y_2$ in the interference channel certainly implies decodability at $Y'_1$ and $Y'_2$ in the channels with disturbance constraint, respectively.
\end{remark}

\begin{figure}[htb]
	\centering
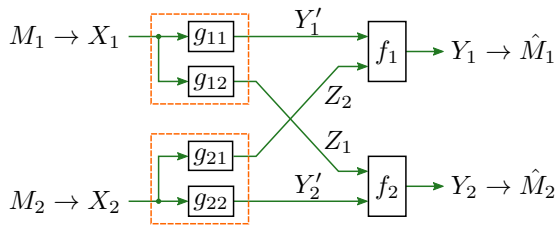%

	\caption{Injective deterministic interference channel with two user pairs. }
	\label{fig:injectiveDet2ic_connection2}
\end{figure}

\begin{example}
Consider the deterministic channel depicted in Figure~\ref{fig:exampleDisturbanceChannel1} and its rate--disturbance region in Figure~\ref{fig:exampleDisturbanceChannel1_region}. Note that rates $R\leq 1$ can be achieved with zero disturbance rate by restricting the transmission to input symbols $\{0,1\}$ (or $\{2,3\}$), which map to different symbols at $Y$, but are indistinguishable at $Z$. On the other hand, for sufficiently large $\Rd$, the disturbance constraint becomes inactive and $R$ is bounded only by the unconstrained capacity $\log(3)$. In addition to the  optimal region achieved by superposition coding, the figure also shows the strictly suboptimal region achieved by simple non-layered random codes.
\end{example}

\begin{figure}[htb]
	\centering
	\subfigure[Channel block diagram]{
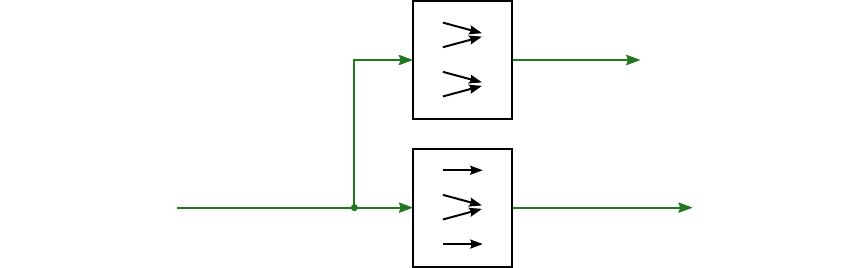%

		\label{fig:exampleDisturbanceChannel1}
	}
	\subfigure[Rate--disturbance region]{
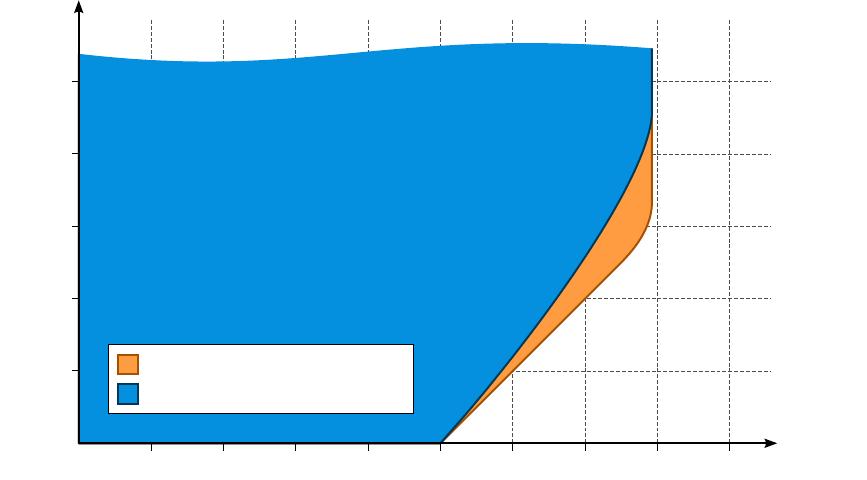%
 
		\label{fig:exampleDisturbanceChannel1_region}
	}
	\caption{Deterministic example with one disturbance constraint. }
	\label{fig:example_1dc}
\end{figure}

\medskip
\subsubsection{Gaussian channel} 

Consider the problem of communication with one disturbance constraint for the Gaussian channel 
\begin{align*}
	Y&=X+ W_1, \\  
	Z&=X+ W_2,
\end{align*}
where the noise is $W_1 \sim \N(0,1)$ and $W_2 \sim \N(0,N)$. Assume an average power constraint $P$ on the transmitted signal $X$. 

The case $N\leq1$ is not interesting, since then $Y$ is a degraded version of $Z$ and the disturbance rate is simply given by the data rate $R$.
If $N>1$, $Z$ is a degraded version of $Y$, and the rate--disturbance region reduces to the following. 

\begin{corollary} \label{thm:dcCapa_Gaussian1}  
	The rate--disturbance region of the Gaussian channel with parameters $P>0$ and $N>1$ is the set of rate pairs $(R,\Rd)$ such that
	\begin{align*} 
		 R  &\leq \C(\alpha P), \\
		 \Rd &\geq \C(\alpha P/N),
	\end{align*}
	for some $\alpha \in [0,1]$, where $\C(x)=(1/2) \log (1 +x)$ for $x \ge 0$.
\end{corollary}
Achievability is proved using Gaussian codes with power $\alpha P$. The converse follows by defining $\alpha^\star \in [0,1]$ such that $R=\C(\alpha^\star P)$ and applying the vector entropy power inequality to $Z^n = Y^n + \tilde W_2^n$, where $\tilde W_2 \sim \N(0,N-1)$ is the excess noise. The details are given in Subsection~\ref{sec:proof_dcCapa_Gaussian1}. Note that this is a degenerate form of the Han--Kobayashi scheme because the constraint from the multiple access side of the interference channel is not taken into consideration.

\medskip
\subsubsection{Vector Gaussian channel}

Now consider the vector Gaussian channel with one disturbance constraint
\begin{align*}
	Y&=X+ W_1, \\ 
	Z&=X+ W_2,
\end{align*}
where $X \in \Real^n$ and the noise $W_1 \sim \N(0,K_1)$ and $W_2 \sim \N(0,K_2)$ for some positive semidefinite covariance matrices $K_1,K_2 \in \Real^{n\times n}$. Assume an average transmit power constraint $\tr(K_x) \leq P$, where $K_x = \Ex{(XX\T)}$ is the covariance matrix of $X$. 
This case is not degraded in general.

\begin{thm} \label{thm:dcCapa_GaussianVector1}  
	The rate--disturbance region of the Gaussian vector channel with parameters $P$, $K_1$, and $K_2$ is the convex hull of the set of pairs $(R,\Rd)$ such that
	\begin{align*} 
		 R  &\leq \tfrac 1 2 \log \frac {\det{K_u + K_v + K_1}}{\det{K_1}}, \\
		 R - \Rd & \leq \tfrac 1 2 \log \frac {\det{K_v+K_1}}{\det{K_v + K_2}} \frac {\det{K_2}}{\det{K_1}} , \\
		 \Rd &\geq \tfrac 1 2 \log \frac {\det{K_v + K_2}}{\det{K_2}}.
	\end{align*}
	for some positive semidefinite matrices $K_u, K_v \in \Real^{n\times n}$ with $\tr(K_u+K_v)\leq P$.  
\end{thm}
Achievability of this rate--disturbance region is shown by applying Theorem~\ref{thm:dcCapa}. Using the discretization procedure in~\cite{ElGamalKim}, it can be shown that the theorem continues to hold with the power constraint additionally applied to the set of permissible input distributions. The claimed region then follows by considering the special case where the input distribution $p(u,x)$ is jointly Gaussian. To prove the converse, we use an extremal inequality in~\cite{Liu2007} to show that Gaussian input distributions are sufficient. The details of the proof are given in Subsection~\ref{sec:proof_dcCapa_GaussianVector1}.

\subsection{Inner and outer bounds for the deterministic channel with two disturbance constraints}
The correspondence between optimal encoding for the channel with one disturbance constraint and the Han--Kobayashi scheme for the interference channel suggests that the optimal coding scheme for $K$ disturbance constraints may provide an efficient (if not optimal) scheme for the interference channel with more than two user pairs. This is particularly the case for extensions of the two user-pair injective deterministic interference channel for which Han--Kobayashi is optimal~\cite{ElGamalCosta82} (see Remark~\ref{rem:dc_ic_correspondence}). As such, we restrict our attention to the deterministic version of the DMC-$2$-DC. 

First, we establish the following inner bound on the rate--disturbance region. 
\begin{thm}[Inner bound] \label{thm:dcCapa_2}
	The rate--disturbance region $\Rr$ of the deterministic channel with two disturbance constraints is inner-bounded by the set of rate triples $(R,\Rdk 1,\Rdk 2)$ such that
	\begin{align}
		 R  &\leq H(Y), \label{eq:dcCapa2_cond1} \\
		 \Rdk 1 + \Rdk 2 &\geq I(Z_1; Z_2\cond U), \label{eq:dcCapa2_cond2} \\
		 R - \Rdk 1 &\leq H(Y\cond Z_1, U), \label{eq:dcCapa2_cond3} \\
		 R - \Rdk 2 &\leq H(Y\cond Z_2, U), \label{eq:dcCapa2_cond4} \\
		 R - \Rdk 1 - \Rdk 2  &\leq H(Y\cond Z_1, Z_2, U)  - I(Z_1; Z_2\cond U), \label{eq:dcCapa2_cond5}\\
		 2R - \Rdk 1 - \Rdk 2  &\leq H(Y\cond Z_1, Z_2, U)  + H(Y\cond U) \notag \\ 
		 & \qquad \qquad- I(Z_1; Z_2\cond U), \label{eq:dcCapa2_cond6}
	\end{align}
	for some pmf $p(u,x)$.
\end{thm}
The inner bound is convex. The expression $I(Z_1; Z_2\cond U)$ appears in three of the inequalities. As in  Marton coding for the 2-receiver broadcast channel with a common message, it is the penalty incurred in encoding independent messages via correlated sequences. The region $\Rr(U,X)$ defined by the inequalities in the theorm for a fixed $p(u,x)$ is illustrated in Figure~\ref{fig:dcCapa2_constituentRegion}. 
\begin{figure}[h]
	\centering
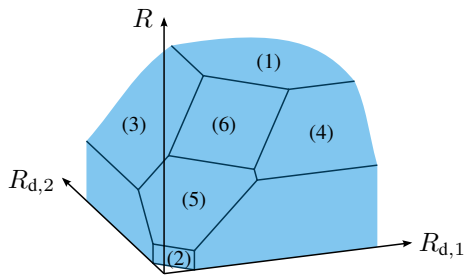%

	\caption{Region $\Rr(U,X)$ for Theorem~\ref{thm:dcCapa_2}. Each face is annotated by the inequality that defines it. }
	\label{fig:dcCapa2_constituentRegion}
\end{figure}


\begin{remark} \label{rem:cond6_alt}
	The right-hand side of condition~\eqref{eq:dcCapa2_cond6} can be equivalently expressed as
	\begin{align*}
		&H(Y\cond Z_1, Z_2, U)  + H(Y\cond U) - I(Z_1; Z_2\cond U) \\
		&\quad \quad =  H(Y\cond Z_1, U)  + H(Y\cond Z_2, U) - I(Z_1; Z_2\cond U,Y),
	\end{align*}
	This shows that the condition is stricter than the sum of conditions~\eqref{eq:dcCapa2_cond3} and~\eqref{eq:dcCapa2_cond4}.
\end{remark}

The encoding scheme for Theorem~\ref{thm:dcCapa_2} involves rate splitting, Marton coding, and superposition coding. The analysis of the probability of error, however, is complicated by the fact that receiver $Y$ wishes to decode all parts of the message as detailed in Subsection~\ref{sec:proof_dcCapa_2}. Receivers $Z_1$ and $Z_2$ each observe a satellite codeword from a superposition codebook. 

\begin{remark}
The encoding scheme underlying the inner bound of Theorem~\ref{thm:dcCapa_2} can be readily extended to the general (non-deterministic) DMC-$2$-DC. 
\end{remark}

To complement the inner bound, we establish the following outer bound on the rate--disturbance region of the deterministic channel with two disturbance constraints.
\begin{thm}[Outer bound] \label{thm:2dc_OuterBound}
If a rate triple $(R,\Rdk 1,\Rdk 2)$ is achievable for the deterministic channel with two disturbance constraints, then it must satisfy the conditions
	\begin{align*}
		R & \leq H(Y\cond Q), \\
		\Rdk 1 &\geq I(Y; Z_1 \cond Q), \\
		\Rdk 2 &\geq I(Y; Z_2 \cond Q), 
	\end{align*}
	for some pmf $p(q,x)$ with $\abs{\Qc} \leq 3$.   
\end{thm}
The proof of this outer bound is given in Subsection~\ref{sec:2dc_OuterBound_Proof}. Note that this outer bound is very similar in form to the alternative description of Corollary~\ref{thm:dcCapa_Det} for the single-constraint deterministic case.

The inner bound in Theorem~\ref{thm:dcCapa_2} and the outer bound in Theorem~\ref{thm:2dc_OuterBound} coincide in some special cases. To discuss these, we introduce the following notation. Since all channel outputs are functions of $X$, they can be equivalently thought of as set partitions of the input alphabet $\Xc$. Set partitions form a partially ordered set (poset) under the refinement relation. Since this poset is a complete lattice~\cite{StanleyEnumerativeCombinatorics}, the following concepts are well-defined. For two set partitions (functions) $f$ and $g$, let $f \refines g$ denote that $f$ is a refinement of $g$ (equivalently, $g$ is degraded with respect to $f$), let $f \wedge g$ be the intersection of the two set partitions (the function that returns both $f$ and $g$), and let $f \vee g$ denote the finest set partition of which both $f$ and $g$ are refinements (the G\'acs--K\"orner--Witsenhausen common part of $f$ and $g$, cf.~\cite{GacsKoerner1973,Witsenhausen1975}).


The inner bound of Theorem~\ref{thm:dcCapa_2} coincides with the outer bound of Theorem~\ref{thm:2dc_OuterBound} if $Z_1$ or $Z_2$ is a degraded version of $Y \wedge (Z_1 \vee Z_2)$, i.e., if the output $Y$ together with the common part of $Z_1$ and $Z_2$ determine $Z_1$ or $Z_2$ completely. 

\begin{thm} \label{thm:dcCapa_2_YZ1Z2}
	The rate--disturbance region $\Rr$ of the deterministic channel with two disturbance constraints is given by the outer bound of Theorem~\ref{thm:2dc_OuterBound} if
	\begin{align*} 
		Y \wedge (Z_1 \vee Z_2) & \refines Z_1, \quad \text{or} \\
		Y \wedge (Z_1 \vee Z_2) & \refines Z_2.
	\end{align*}
\end{thm}
The theorem is proved by specializing Theorem~\ref{thm:dcCapa_2} as detailed in Subsection~\ref{sec:proof_dcCapa_2_YZ1Z2}. In the case where $Z_1$ or $Z_2$ is a degraded version of $Y$ alone, achievability follows by setting $U=\emptyset$ in Theorem~\ref{thm:dcCapa_2}. Otherwise, we let $U= Z_1 \vee Z_2$. This is intuitive, since $U$ corresponds to the common-message step in the Marton encoding scheme.

\begin{example} \label{ex:2dc} 
Consider the deterministic channel depicted in Figure~\ref{fig:exampleDisturbanceChannel2}. The desired receiver output $Y$ is a refinement of both side receiver outputs $Z_1$ and $Z_2$, and hence, Theorem~\ref{thm:dcCapa_2_YZ1Z2} applies.  Figure~\ref{fig:exampleDisturbanceChannel2_region} depicts the rate--disturbance region, numerically approximated by evaluating each grid point in a regular grid over the distributions $p(x)$ and subsequently taking the convex hull. Figure~\ref{fig:exampleDisturbanceChannel2_cut1} contrasts the single-constraint case (where $\Rdk 2$ is set to infinity, and thus inactive) with the case where both side receivers are under the same disturbance rate constraint ($\Rdk 1 = \Rdk 2$). As expected, imposing an additional disturbance constraint can significantly reduce the achievable message rate. Finally, Figure~\ref{fig:exampleDisturbanceChannel2_cut2} illustrates the trade-off between the disturbance rates $\Rdk1$ and $\Rdk2$ at the two side receivers, for a fixed data rate $R$. 
\end{example}

\begin{figure}[htb]
	\centering
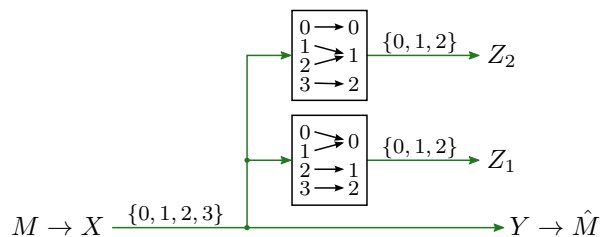%

	\caption{Deterministic channel with two disturbance constraints (Example~\ref{ex:2dc}).}
	\label{fig:exampleDisturbanceChannel2}
\end{figure}

\begin{figure}[htb] 
	\centering
	\subfigure[Rate--disturbance region]{
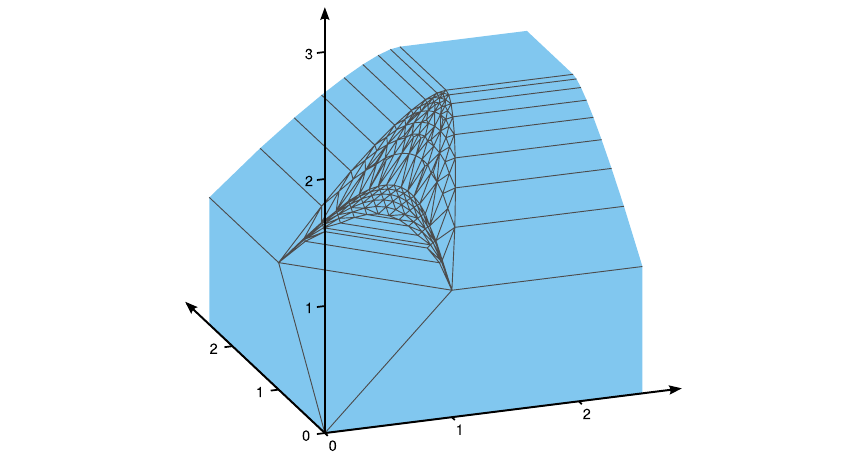%
 
		\label{fig:exampleDisturbanceChannel2_region}
	}
	\subfigure[Single disturbance constraint ($\Rdk 1 = \Rd$, $\Rdk 2 = \infty$) and symmetric disturbance constraint ($\Rdk 1 = \Rdk 2 = \Rd$).]{
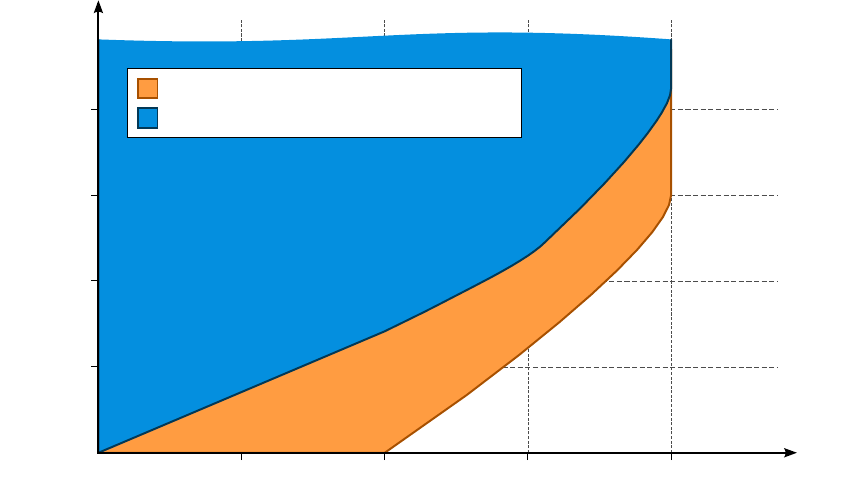%
 
		\label{fig:exampleDisturbanceChannel2_cut1}
	}
	\subfigure[Contour lines of the rate--disturbance region at constant rate $R$.]{
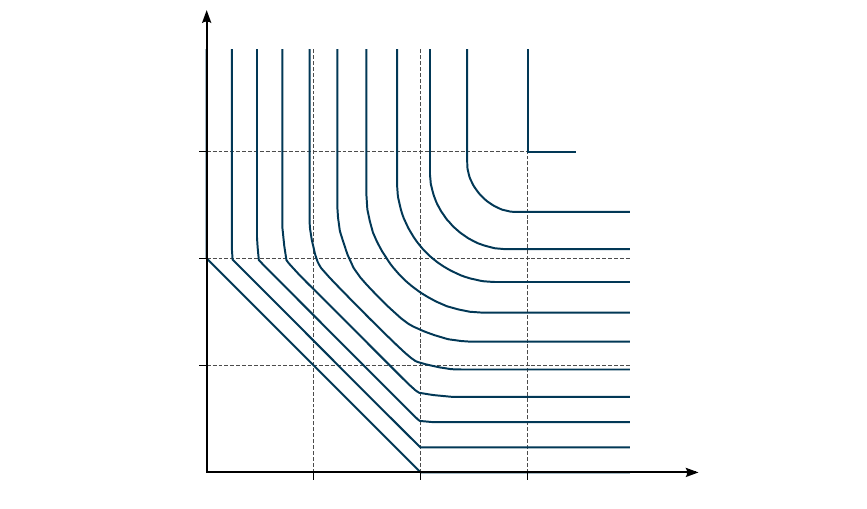%
 
		\label{fig:exampleDisturbanceChannel2_cut2}
	}
	\caption{Rate--disturbance region for Example~\ref{ex:2dc}. }
	\label{fig:example_2dc}
\end{figure}

We conclude this section by considering another case in which we can fully characterize the rate--disturbance region of the deterministic channel with two disturbance constraints. If $Z_1$ is a degraded version of $Z_2$ (or vice versa), the region $\Rr$ of Theorem~\ref{thm:dcCapa_2} is optimal and simplifies to the following.

\begin{corollary} \label{thm:dcCapa_2_degraded}
	The rate--disturbance region $\Rr$ of the deterministic channel with two disturbance constraints with $Z_1 \refines Z_2$ or $Z_2 \refines Z_1$ is the set of rate triples $(R,\Rdk 1,\Rdk 2)$ such that
	\begin{align*}
		 R  &\leq H(Y), \\ 
		 R - \Rdk 1 &\leq H(Y\cond Z_1), \\
		 R - \Rdk 2 &\leq H(Y\cond Z_2). 
	\end{align*}
	for some pmf $p(x)$.
\end{corollary}
Achievability follows as a special case of Theorem~\ref{thm:dcCapa_2}. The encoding scheme underlying the theorem carefully avoids introducing an ordering between the side receiver signals $Z_1$ and $Z_2$, but such ordering is naturally given by the channel here. Consequently, the corollary follows by setting the auxiliary $U$ equal to the output at the degraded side receiver. This turns the encoding scheme into superposition coding with three layers. The details are given in Subsection~\ref{sec:proof_dcCapa_2_degraded}. 

Note that the region of Corollary~\ref{thm:dcCapa_2_degraded} is akin to the deterministic case with one disturbance constraint in Corollary~\ref{thm:dcCapa_Det}. In both cases, the side receiver signals need not be degraded with respect to $Y$.

\clearpage
\section{Proofs for a single disturbance constraint}  \label{sec:proofs_single}
\subsection{Achievability proof of Theorem~\ref{thm:dcCapa}}   \label{sec:achievability_dcCapa}
Achievability is proved as follows. 

\vspace{2mm}
\noindent \emph{Codebook generation.}  Fix a pmf $p(u,x)$. 
\begin{enumerate}
	\item Split the message $M$ into two independent messages $M_0$ and $M_1$ with rates $R_0$ and $R_1$, respectively. Hence $R=R_0+R_1$.
	\item For each $m_0 \in \natSet{2^{nR_0}}$, independently generate a sequence $u^n(m_0)$ according to $\prod_{i=1}^n p(u_i)$.
	\item For each $(m_0,m_1) \in \natSet{2^{nR_0}}\times \natSet{2^{nR_1}}$, independently generate a sequence $x^n(m_0,m_1)$ according to $\prod_{i=1}^n p(x_i\cond u_i(m_0))$.
\end{enumerate}

\vspace{2mm}
\noindent \emph{Encoding.}  To send message $m=(m_0,m_1)$, transmit $x^n(m_0,m_1)$.

\vspace{2mm}
\noindent \emph{Decoding.} Upon receiving $y^n$, find the unique $(\hat m_0, \hat m_1)$ such that 
$(u^n(\hat m_0), x^n(\hat m_0, \hat m_1), y^n) \in \Typ(U,X,Y)$.

\vspace{2mm}
\noindent \emph{Analysis of the probability of error.} We are using a superposition code over the channel from $X$ to $Y$. Using the law of large numbers and the packing lemma in~\cite{ElGamalKim}, it can be shown that the probability of error tends to zero as $n \to \infty$ if
\begin{align}
	R_1 & < I(X;Y \cond U) - \delta(\eps),     \label{eq:R1_bound} \\
	R_0+R_1 & < I(X;Y) - \delta(\eps).    \label{eq:R0R1_bound}
\end{align}

\vspace{2mm}
\noindent \emph{Analysis of disturbance rate.} We analyze the disturbance rate averaged over codebooks $\Cc$.
\begin{align}
	I(X^n;Z^n\cond \Cc) &\leq H(Z^n, M_0\cond \Cc) - H(Z^n\cond X^n,\Cc) \notag \\
	&= H(M_0) + H(Z^n \cond M_0,\Cc) - H(Z^n\cond X^n) \notag \\
	&\annleq{a} nR_0 + H(Z^n \cond U^n) - nH(Z\cond X)  \notag  \\
	&\leq nR_0 + nH(Z\cond U) - nH(Z\cond X,U) \notag \\
	&= nR_0 + nI(X;Z\cond U) \notag \\
	&\leq n \Rd,  \label{eq:Rd_bound}
\end{align}
where (a) follows since $U^n$ is a function of the codebook $\Cc$  and $M_0$. Substituting $R=R_0+R_1$ and using Fourier--Motzkin elimination on inequalities~\eqref{eq:R1_bound},~\eqref{eq:R0R1_bound}, and~\eqref{eq:Rd_bound} completes the proof of achievability.

\subsection{Converse of Theorem~\ref{thm:dcCapa}}   \label{sec:converse_dcCapa}   

Consider a sequence of codes with $\pen \to 0$ as $n \to \infty$ and the joint pmf that it induces on $(M,X^n,Y^n,Z^n)$ assuming $M\sim \mathrm{Unif}\natSet{2^{nR}}$. Define the time-sharing random variable $Q \sim \mathrm{Unif}\natSet n$, independent of everything else. We use the identification $U = (Q,Y_{Q+1}^n,Z^{Q-1} )$, and let $X = X_Q$, $Y=Y_Q$, and $Z=Z_Q$. Note that $(X,Y,Z)$ is consistent with the channel. Then 
\begin{align*}
	R &\leq I(X;Y) + \eps_n,
\end{align*}
as in the converse proof for point-to-point channel capacity, which uses the same identifications of random variables.
On the other hand,
\begin{align*}
	n \Rd & \geq I(X^n;Z^n) \\
	&= H(Z^n) - H(Z^n\cond X^n) \\
	&= \sum_{i=1}^n \left( H(Z_i \cond Z^{i-1} ) - H(Z_i\cond X_i) \right) \\
	&\geq \sum_{i=1}^n H(Z_i \cond Z^{i-1}, Y_{i+1}^n ) - nH(Z\cond X) \\
	& = nH(Z \cond U ) - nH(Z\cond X,U) \\
	& = nI(X;Z\cond U).
\end{align*}
Finally,
{\allowdisplaybreaks
\begin{align*}  
	&n(\Rd - R) \\
	&\geq I(X^n;Z^n) - nR \\
	&\anngeq{a} H(Z^n) - H(Z^n\cond X^n) - I(M;Y^n) - n\eps_n  \\
	&\anneq{b} \sum_{i=1}^n \left( H(Z_i \cond Z^{i-1}) - I(M;Y_i\cond Y_{i+1}^n) \right) - nH(Z\cond X) - n\eps_n \\
	&= \sum_{i=1}^n \left( H(Z_i \cond Z^{i-1}, Y_{i+1}^n) + I(Y_{i+1}^n; Z_i \cond Z^{i-1}) \right. \\*
	&\qquad \left. - H(Y_i\cond Y_{i+1}^n) +H(Y_i\cond M, Y_{i+1}^n) \right) - nH(Z\cond X) - n\eps_n \\
	&\anneq{c} \sum_{i=1}^n \left( H(Z_i \cond Z^{i-1}, Y_{i+1}^n) + I(Y_i; Z^{i-1} \cond Y_{i+1}^n) \right. \\*
	&\qquad \left. - H(Y_i\cond Y_{i+1}^n) +H(Y_i\cond X_i) \right) - nH(Z\cond X) - n\eps_n \\
	&= \sum_{i=1}^n \left( H(Z_i \cond Z^{i-1}, Y_{i+1}^n) - H(Y_i \cond Z^{i-1}, Y_{i+1}^n) \right. \\*
	&\qquad \left. +H(Y_i\cond X_i,Z^{i-1}, Y_{i+1}^n) \right) - nH(Z\cond X) - n\eps_n \\
	&= \sum_{i=1}^n \left( H(Z_i \cond Z^{i-1}, Y_{i+1}^n) - I(X_i; Y_i \cond Z^{i-1}, Y_{i+1}^n)  \right) \\*
	&\qquad - nH(Z\cond X) - n\eps_n \\
	&\anneq{d} nH(Z\cond U) - nI(X;Y\cond U) - nH(Z\cond X,U)  - n\eps_n \\
	&= nI(X;Z\cond U) - I(X;Y\cond U) - n\eps_n,
\end{align*}
}%
where (a) uses Fano's inequality, (b) single-letterizes the noise term $H(Z^n\cond X^n)$ with equality due to memorylessness of the channel,  (c) applies Csisz\'ar's sum identity to the second term and channel memorylessness to the fourth term, and (d) uses the previous definitions of auxiliary random variables. Finally, the cardinality bound on $\Uc$  is established using the convex cover method in~\cite{ElGamalKim}.

\subsection{Proof of Corollary~\ref{thm:dcCapa_Det}}   \label{sec:proof_dcCapa_Det}
Using the deterministic nature of the channel, the region in  Theorem~\ref{thm:dcCapa} reduces to the set of rate pairs $(R,\Rd)$ such that 
\begin{align}
	 R  &\leq H(Y ),  
	 \\
	 \Rd &\geq H(Z\cond U),  \label{eq:detRd1} \\
	 \Rd &\geq R+ H(Z\cond U) - H(Y\cond U), \label{eq:detRd2} 
\end{align}
for some pmf $p(u,x)$. Now fixing a rate $R$ and a pmf $p(x)$ and varying $p(u|x)$ to minimize $\Rd$, the right hand sides of~\eqref{eq:detRd1} and~\eqref{eq:detRd2} are lower bounded by
\begin{align*}
	& H(Z\cond U) \geq 0, 
\end{align*}
and
\begin{align*}
	&R+ H(Z\cond U) - H(Y\cond U) \\
	& = R+ H(Z\cond U) - H(Y,Z\cond U) + H(Z\cond Y,U) \\
	&= R-H(Y\cond Z,U) + H(Z\cond Y,U) \\
	&\geq R- H(Y\cond Z).
\end{align*}
Note that the particular choice $U=Z$ simultaneously achieves both lower bounds with equality and is therefore sufficient. The rate--disturbance region thus reduces to Corollary~\ref{thm:dcCapa_Det}.

For a fixed pmf $p(x)$, this region has exactly two corner points: $P_1 = (H(Y|Z),0)$ and $P_2 = (H(Y), I(Y;Z))$. As we vary $p(x)$, there is one corner point $P_1$ that dominates all other $P_1$ points. The pmf $p(x)$ for this dominant $P_1$ can be constructed by maximizing $H(Y|Z)$ as follows. For each $z\in \Zc$, define $\Yc_z \subseteq \Yc$ to be the set of $y$ symbols that are compatible with $z$. Let $z^\star$ be a symbol that maximizes $\card{\Yc_z}$. For each element of $\Yc_{z^\star}$, pick exactly one $x$ that is compatible with it and $z^\star$. Finally, place equal probability mass on each of these $x$ values, and zero mass on all others. This pmf on $X$ yields the dominant corner point $P_1$, namely $( \log(\card{\Yc_{z^\star}}),0)$. Moreover, for this distribution, $P_2$ coincides with $P_1$. Therefore, the net contribution (modulo convexification) of each pmf $p(x)$ to the rate--disturbance region amounts to its corner point $P_2$. This implies the alternative description of the region. Lastly, the cardinality bound on $\Qc$ in the alternative description is follows from the convex cover method in~\cite{ElGamalKim}.

\subsection{Proof of Corollary~\ref{thm:dcCapa_Gaussian1}}   \label{sec:proof_dcCapa_Gaussian1}  
Achievability is straightforward using a random Gaussian codebook with power control, and upper-bounding the disturbance rate at receiver $Z$ by white Gaussian noise. The converse can be seen as follows.
Clearly, $R \leq \C(P)$. Let $\alpha^\star \in [0,1]$ be such that $R=\C(\alpha^\star P)$. Then
	\begin{align*}
		n\C(\alpha^\star P) 
		= nR & \leq I(X^n;Y^n) + n \eps_n \\
		&= h(Y^n) - h(Y^n\cond X^n) + n \eps_n,
	\end{align*}
and therefore,
	\begin{align*}
		h(Y^n) &\geq \tfrac n 2 \log(2\pi e) + n\C(\alpha^\star P) - n \eps_n \\
		&= \tfrac n 2 \log \left(2\pi e(1+\alpha^\star P)\right) - n \eps_n
	\end{align*}
	Since $N<1$, we can write the physically degraded form of the channel as $Y=X+W_1$, $Z=Y+\tilde W_2$, where $\tilde W_2 \sim \N(0,N-1)$ is the excess noise that receiver $Z$ experiences in addition to receiver $Y$. Applying the vector entropy power inequality to $Z^n = Y^n + \tilde W_2^n$, we conclude
	\begin{align*}
		\tfrac 1 n h(Z^n) 
		& \geq \tfrac 1 2 \log \left( 2^{\frac 2 n h(Y^n)} + 2^{\frac 2 n h(\tilde W_2^n)} \right) \\
		& \geq \tfrac 1 2 \log \left( 2^{-2\eps_n} \cdot 2\pi e (1+\alpha^\star P) + 2 \pi e (N-1) \right) \\
		& \geq \tfrac 1 2 \log \left( 2\pi e (N+\alpha^\star P) \right) - \eps_n,
	\end{align*}
	and finally,
	\begin{align*}
		\Rd &\geq \tfrac 1 n I(X^n;Z^n) \\
		&= \tfrac 1 n h(Z^n) - \tfrac 1 2 \log( 2\pi e N) \\
		&\geq \C(\alpha^\star P/N) - \eps_n.
	\end{align*}

\subsection{Proof of Theorem~\ref{thm:dcCapa_GaussianVector1}} \label{sec:proof_dcCapa_GaussianVector1}
Recall the shape of $\Rr(U,X)$ depicted in Figure~\ref{fig:dcCapa_constituentRegion}. The coordinates of the corner points $A$ and $B$ are given by
	\begin{align}
		A(U,X): \quad 
			R &= h(X+W_1) - h(W_1), \label{eq:A_R} \\
			\Rd &= h(X+W_2\cond U) + h(X+W_1) - h(X+W_1\cond U) - h(W_2), \label{eq:A_Rd} \\
		B(U,X): \quad 
			R &= h(X+W_1 \cond U) - h(W_1), \label{eq:B_R} \\
			\Rd &= h(X+W_2\cond U)  - h(W_2). \label{eq:B_Rd}
	\end{align}

\begin{proof}[Proof of achievability]
	We specialize Theorem~\ref{thm:dcCapa}. Consider the specific $p(u,x)$ constructed as follows. For given positive semidefinite matrices $K_u, K_v \in \Real^{n\times n}$ with $\tr(K_u+K_v)\leq P$, let
	\begin{align*}
		U & \sim \N(0, K_u), \\
		V & \sim \N(0, K_v), \\
		X & = U+V,
	\end{align*}
	where $U$ and $V$ are independent. Then, the terms in Theorem~\ref{thm:dcCapa} evaluate to 
	\begin{align*}
		I(X;Y) &= h(Y) - h(W_1) = \tfrac 1 2 \log \frac {\det{K_u+K_v+K_1}}{\det{K_1}}, \\
		I(X;Y\cond U) &= h(Y\cond U) - h(W_1) = \tfrac 1 2 \log \frac {\det{K_v+K_1}}{\det{K_1}}, \\
		I(X;Z\cond U) &= h(Z\cond U) - h(W_2) = \tfrac 1 2 \log \frac {\det{K_v+K_2}}{\det{K_2}}.
	\end{align*}
	Simplifying the right hand sides and introducing time-sharing leads to the desired result.
	
	For completeness, the coordinates of $A$ and $B$ for given matrices $K_u$, $K_v$ are
	\begin{align}
		A(K_u,K_v): \quad 
			R &= \tfrac 1 2 \log \frac{\det{K_u+K_v+K_1}}{\det{K_1}}, \label{eq:AKuKv_R} \\
			\Rd &= \tfrac 1 2 \log \frac{\det{K_v+K_2}}{\det{K_2}} \frac{\det{K_u+K_v+K_1}}{\det{K_v+K_1}}, \label{eq:AKuKv_Rd} \\
		B(K_u,K_v): \quad 
			R &= \tfrac 1 2 \log \frac{\det{K_v+K_1}}{\det{K_1}}, \label{eq:BKuKv_R} \\
			\Rd &= \tfrac 1 2 \log \frac{\det{K_v+K_2}}{\det{K_2}}. \label{eq:BKuKv_Rd} 
	\end{align}
	The constituent region $\Rr(U,X)$ for fixed $K_u$ and $K_v$ is depicted in Figure~\ref{fig:dcCapa_constituentRegion_Gaussian}.
	\begin{figure}[htb]
		\centering
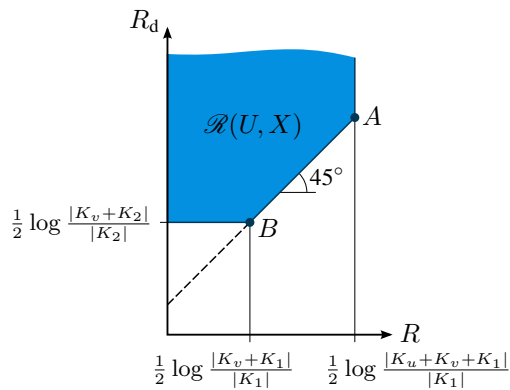%

		\caption{Constituent region for Gaussian superposition codebook with parameters $K_u$ and $K_v$.}
		\label{fig:dcCapa_constituentRegion_Gaussian}
	\end{figure}
\end{proof}

\begin{proof}[Proof of converse]
	The converse proof of Theorem~\ref{thm:dcCapa} continues to hold and we only need to show that Gaussian input distributions are sufficient. We proceed as follows. Since the rate--disturbance region is convex, its boundary can be fully characterized by maximizing $R - \lambda \Rd$ for each $\lambda>0$. We write
	\begin{align}
		R-\lambda\Rd &\leq \max_{(R,\Rd) \in \Rr} \left\{ R - \lambda \Rd \right\} \notag \\
		&= \max_{(U,X)} \max_{(R,\Rd) \in \Rr(U,X)} \left\{ R - \lambda \Rd \right\}, \notag 
	\end{align}
	where the outer optimization is over the joint distribution of $(U,X)$ and the inner optimization is over the region achieved by that distribution. The inner optimization can be solved explicitly as follows. For ease of presentation, assume for the moment that the power constraint is of the form $K_x \preceq S$ for some positive semidefinite matrix $S$. (That is, valid $K_x$ are precisely those that result in the matrix $S-K_x$ being positive semidefinite.)
	
	First, consider $\lambda \leq 1$. For any distribution $(U,X)\sim p(u,x)$, point $A(U,X)$ achieves a value of the inner optimization at least as large as point $B(U,X)$, or any point on the line between them. Using the coordinates of $A(U,X)$ in~\eqref{eq:A_R} and~\eqref{eq:A_Rd}, we can write
	\begin{align*}
		R-\lambda\Rd & \leq \max_{(U,X)} \left\{ \lambda\left( h(X+W_1\cond U) - h(X+W_2\cond U) \right) \right. \\
		& \qquad \qquad  + \left. (1-\lambda) h(X+W_1) - h(W_1) + \lambda h(W_2) \right\} \\
		& \annleq{a} \lambda \cdot \max_{(U,X)} \left\{ h(X+W_1\cond U) - h(X+W_2\cond U) \right\} \\
		& \quad + (1-\lambda) \cdot \max_{(U,X)} \left\{h(X+W_1) \right\} - h(W_1) + \lambda h(W_2) \\
		& \annleq{b} \lambda \cdot \max_{K_x \preceq S} \left\{ \tfrac 1 2 \log \frac {\det{K_x+K_1}}{\det{K_x+K_2}} \right\} + (1-\lambda) \cdot \max_{K_x\preceq S} \left\{ \tfrac 1 2 \log \left((2\pi e)^n \det{K_x+K_1}\right) \right\} \\
		& \quad - \tfrac 1 2 \log \left( (2\pi e)^n \det{K_1} \right) + \tfrac \lambda 2 \log \left( (2\pi e)^n \det{K_2} \right) .
	\end{align*}
	In (a), the two maximizations are taken independently. In step (b), the first maximization is achieved by a Gaussian $X$ that is independent of $U$, due to a theorem proved by Liu and Viswanath~\cite[Thm. 8]{Liu2007}. The optimization is now only over covariances matrices. Let $K^\star$ be an optimizer of this first maximization. The second maximization is also achieved by a Gaussian $X$, and is optimized by $K_x=S$ since $f(K_x) =  \det{K_x+K_1}$ is matrix monotone.  It follows that
	\begin{align*}
		R-\lambda\Rd 
		& \leq \tfrac \lambda 2 \log \frac {\det{K^\star+K_1}}{\det{K^\star+K_2}}  + \tfrac {1-\lambda} 2 \log \left((2\pi e)^n \det{S+K_1}\right)  \\
		& \quad - \tfrac 1 2 \log \left( (2\pi e)^n \det{K_1} \right) + \tfrac \lambda 2 \log \left( (2\pi e)^n \det{K_2} \right)  \\
		& = \tfrac 1 2 \log \frac {\det{S+K_1}}{\det{K_1}} -
		\tfrac \lambda 2 \log \frac {\det{K^\star+K_2}}{\det{K^\star + K_1}} \frac {\det{S+K_1}}{\det{K_2}}.
	\end{align*}
	But this upper bound is achieved with equality by Gaussian superposition codebooks, namely through the point $A(K_u,K_v)$ as specified by equations~\eqref{eq:AKuKv_R} and~\eqref{eq:AKuKv_Rd}, with $K_u = S - K^\star$ and $K_v = K^\star$.
	
	Now, consider $\lambda>1$. The argument proceeds analogously to the previous case. For completeness' sake, the details are as follows. We can write the inner optimization explicitly using the coordinates of $B(U,X)$ in~\eqref{eq:B_R} and~\eqref{eq:B_Rd} as 
	\begin{align*}
		R-\lambda\Rd & \leq \max_{(U,X)} \left\{ h(X+W_1\cond U) - \lambda h(X+W_2\cond U)\right\} + \lambda h(W_2) - h(W_1)  \\
		&\annleq{a} \max_{K_x \preceq S} \left \{ \tfrac 1 2 \log \left((2\pi e)^n\det{K_x+K_1} \right) - \tfrac \lambda 2 \log \left((2\pi e)^n\det{K_x+K_2} \right) \right\} \\
		& \quad + \tfrac \lambda 2 \log \left((2\pi e)^n\det{K_2} \right) - \tfrac 1 2 \log \left((2\pi e)^n\det{K_1} \right).
	\end{align*}
	The optimum in (a) is achieved by a Gaussian $X$ (independent of $U$) by virtue of~\cite[Thm. 8]{Liu2007}, while the other two terms are independent of the optimization variable. Let $K^\star$ be an optimizer. Then 
	\begin{align*}
		R-\lambda\Rd & \leq \tfrac 1 2 \log \frac{\det{K^\star+K_1}}{\det{K_1}} - \tfrac \lambda 2 \log  \frac{\det{K^\star+K_2}}{\det{K_2}}.
	\end{align*}
	This upper bound is achieved with equality by Gaussian superposition codebooks through the point $B(K_u, K_v)$ as given by equations~\eqref{eq:BKuKv_R} and~\eqref{eq:BKuKv_Rd} with $K_u = 0$ and $K_v = K^\star$. This is a power control strategy, similar to the scalar Gaussian case.
	
	We have thus shown that under a power constraint $K_x \preceq S$, Gaussian superposition codes are optimal. The conclusion extends to the sum power constraint $\tr(K_x) \leq P$ by observing that
	\begin{align*}
		 \{K_x: \tr(K_x) \leq P \} &= \bigcup_{\substack{S:\,S\succeq 0\\ \tr(S)\leq P}} \{K_x: K_x \preceq S\}.
	\end{align*}
	In other words, the sum power constraint can be expressed as a union of constraints of the type $K_x \preceq S$, for each of which Gaussian superposition codes are optimal. Therefore, a Gaussian superposition code must be optimal overall, too.
\end{proof}

\section{Proofs for two disturbance constraints}  \label{sec:proofs_two}
\subsection{Proof of Theorem~\ref{thm:dcCapa_2}} \label{sec:proof_dcCapa_2}
\noindent \emph{Codebook generation.}  Fix a pmf $p(u,x)$. Split the rate as $R=R_0+R_1+R_2+R_3$. Define the auxiliary rates $\tilde R_1 \geq R_1$ and $\tilde R_2 \geq R_2$, let $\eps'>0$, and define the set  partitions 
\begin{align*}
	\natSet{2^{n\tilde R_1}} &= \Lc_{1}(1) \cup \dots \cup \Lc_{1}(2^{nR_1}), \\
	\natSet{2^{n\tilde R_2}} &= \Lc_{2}(1) \cup \dots \cup \Lc_{2}(2^{nR_2}),
\end{align*}
where $\Lc_1(\cdot)$ and $\Lc_2(\cdot)$ are indexed sets of size $2^{n(\tilde R_1-R_1)}$ and $2^{n(\tilde R_2-R_2)}$, respectively. 
\begin{enumerate}
	\item For each $m_{0} \in \natSet{2^{nR_0}}$, generate $u^n(m_{0})$ according to $\prod_{i=1}^n p(u_i)$.
	\item For each $l_{1}\in \natSet{2^{n\tilde R_{1}}}$, generate $z^n_{1}(m_0,l_{1})$ according to $\prod_{i=1}^n p(z_{1i}\cond u_i(m_0))$. Likewise, for each $l_{2}\in \natSet{2^{n\tilde R_{2}}}$, generate $z^n_{2}(m_0, l_{2})$ according to $\prod_{i=1}^n p(z_{2i}\cond u_i(m_0))$. \label{item:MC_Gen_indep1}
	\item For each $ (m_{0},m_{1},m_{2})$, let $\Sc(m_0,m_1,m_2)$ be the set of all pairs $(l_{1},l_{2})$ from the product set $\Lc_{1}(m_{1}) \times \Lc_{2}(m_{2})$ such that $( z^n_{1}(m_0,l_{1}), z^n_{2}(m_0,l_{2}) ) \in \Typprime(Z_{1}, Z_{2} \cond u^n(m_{0}))$. 
	\item For each $(m_{0}, l_{1}, l_{2})$ and $m_{3} \in \natSet{2^{nR_{3}}} $, generate
			$x^n(m_{0}, l_1,l_2,m_3) $ according to 
			\begin{align*}
				\prod_{i=1}^n p ( x_{i} \cond u_i(m_{0}), z_{1i}(l_{1}), z_{2i}(l_{2}) )
			\end{align*}
		if $(l_1,l_2) \in \Sc(m_0,m_1,m_2)$. Otherwise, we draw from $\mathrm{Unif}(\Xc^n)$.
	\label{item:MC_Gen_superpose}
	\item Choose $(l_{1}^{(m_{0},m_{1},m_{2})} ,  l_{2}^{(m_{0},m_{1},m_{2})})$ uniformly from $\Sc(m_0,m_1,m_2)$. If $\Sc(m_0,m_1,m_2)$ is empty, choose $(1,1)$.  \label{item:MC_Gen_Error}		
\end{enumerate}

\vspace{2mm}
\noindent \emph{Encoding.}  To send message $m=(m_0,m_1,m_2,m_3)$, transmit the sequence $$x^n(m_0, l_{1}^{(m_0,m_{1},m_{2})}, l_{2}^{(m_0,m_{1},m_{2})}, m_3).$$ 

\vspace{2mm}
\noindent \emph{Decoding.} Let $\eps>\eps'$. Upon receiving $y^n$, define the tuple
\begin{align*}
	& T(m_0,m_1,m_2,m_3)  \\
	&= \left( 
		u^n(m_0), 
		z^n_{1}(m_0, l_{1}^{(m_0, m_{1}, m_{2} )}), 
		z^n_{2}(m_0, l_{2}^{(m_0, m_{1}, m_{2} )}),  \right. \\
		& \hspace{21mm} \left. x^n(m_0, l_{1}^{(m_0, m_{1}, m_{2} )}, l_{2}^{(m_0, m_{1}, m_{2} )}, m_3),  
		y^n 
	\right)
\end{align*}
Declare that $\hat m = (\hat m_0, \hat m_1, \hat m_2, \hat m_3)$ has been sent if it is the unique message such that $$ T(\hat m_0,\hat m_1,\hat m_2,\hat m_3) \in \Typ(U,Z_1,Z_2,X,Y).$$

\vspace{2mm}
\noindent \emph{Analysis of the probability of error.} Without loss of generality, assume that $m_0=m_1=m_2=m_3=1$ is transmitted. Define the following events.
\begin{align*}
	\Er_\text{e1}: & \quad \text{$\Sc(1,1,1)$ is empty}, \\
	\Er_\text{e2}: & \quad \text{$\Sc(1,1,1)$ contains two distinct pairs with} \\
		&\quad \qquad \text{equal first or second component},  \\
	\Er_i: & \quad \{ T(m_0,m_1,m_2,m_3) \in \Typ(U,Z_1,Z_2,X,Y) \text{ for } \\
		&\quad \ \, \text{some $(m_0,m_1,m_2,m_3) \in \Mc_i$} \}, \qquad\qquad i \in \{0,\dots,5\},
\end{align*}
where the message subsets $\Mc_i$ are specified in Table~\ref{tab:MC_error_cases}. Defining the ``encoding error'' event $\Er_\text{e} = \Er_\text{e1} \cup \Er_\text{e2}$ and the ``decoding error'' event $\Er_\text{d} = \Er_0^c \cup \Er_1 \cup \Er_2 \cup \Er_3 \cup \Er_4 \cup \Er_5$, the probability of error can be upper-bounded as
\begin{align*}
	\P(\Er) 
	&\leq \P( \Er_\text{e} \cup \Er_\text{d} ) 
	\leq \P( \Er_\text{e} ) + \P( \Er_\text{d} \cond \Er_\text{e}^c ) .
\end{align*}
The motivation for introducing $\Er_\text{e2}$ as an ``error'' is to simplify the analysis of the second probability term. 

We bound $\P( \Er_\text{e} )$ by the following lemma. Let $r_{1} = \tilde R_{1} - R_{1}$ and $r_{2}= \tilde R_{2} - R_{2}$.

\begin{tabelle}
	\begin{center}
	{\small
		\begin{tabular}{|c||c|c|c|c|}   
			\hline
			Message subset & $m_{0}$& $m_{1}$ & $m_{2}$ & $m_{3}$ 
			 \\ \hline \hline
			$\Mc_0$ & $\eeq 1$ & $\eeq 1$ & $\eeq 1$& $\eeq 1$ \\ \hline
			$\Mc_1$ & $\eeq 1$ & $\eeq 1$& $\eeq 1$ & $\neq 1$    \\
			$\Mc_2$ & $\eeq 1$ & $\neq 1$& $\eeq 1$ & any   \\
			$\Mc_3$ & $\eeq 1$ & $\eeq 1$& $\neq 1$ & any   \\
			$\Mc_4$ & $\eeq 1$ & $\neq 1$& $\neq 1$ & any   \\ 
			$\Mc_5$ & $\neq 1$ & any & any & any   \\ \hline
		\end{tabular}
	}
	\end{center}
	\caption{Message subsets for decoding error events.}
	\label{tab:MC_error_cases}
\end{tabelle}

\begin{lemma}\label{lemma:encodingError}
	$\P(\Er_\text{e}) \to 0$ as $n\to\infty$ if 
	\begin{align}
		r_{1} + r_{2} &> I(Z_1; Z_1 \cond U) + \delta(\eps')  \label{eq:atLeastOne}, \\
		r_{1}/2 + r_{2} &< I(Z_1; Z_2 \cond U) - \delta(\eps') \label{eq:noDoubles1}, \\
		r_{1} + r_{2}/2 &< I(Z_1; Z_2 \cond U) - \delta(\eps')  \label{eq:noDoubles2}.
	\end{align}
\end{lemma}

\begin{proof}[Proof sketch]
First, consider $\Er_\text{e1}$. As in the proof of Marton's inner bound for the broadcast channel, the mutual covering lemma~\cite{ElGamalKim} implies $\P(\Er_\text{e1}) \to 0$ as $n\to \infty$ if~\eqref{eq:atLeastOne} holds.

Now consider $\Er_\text{e2}$, for which we need to control the number of typical pairs that can occur in the same ``row'' or ``column'' of the product set $\Lc_1(m_1) \times \Lc_2(m_2)$, i.e., for the same $l_{1}$ or $l_{2}$ coordinate.  The probability $\P(\Er_\text{e2})$ tends to zero provided that~\eqref{eq:noDoubles1} and~\eqref{eq:noDoubles2} hold. 

This is akin to the birthday problem~\cite{Mises1939}, where $k$ samples are drawn uniformly and independently from $\natSet N$, and the interest is in samples that have the same value (collisions). It is well-known that for the probability of collision to be $p_c$, the number of samples required is roughly $k \approx \sqrt{-2 N \ln(1-p_c)  }$, which scales with $\sqrt{N}$. 
In our case, the number of samples is the cardinality of the set $\Sc(m_0,m_1,m_2)$, which is roughly $k = 2^{n(r_{1} + r_{2}-I(Z_1; Z_2\cond U))}$. The samples are categorized into $N_1 = 2^{nr_{1}}$ and $N_2 = 2^{nr_{2}}$ classes along rows and columns, respectively. To achieve a probability of collision $p_c \to 0$ along both dimensions, we need $k \ll \min\{\sqrt{N_1},\sqrt{N_2}\}$, which yields exactly the conditions~\eqref{eq:noDoubles1} and~\eqref{eq:noDoubles2}. 

A rigorous proof is given in Appendix~\ref{sec:encodingErrorProof}.
\end{proof}

Before we proceed to bound the probability of decoding error, we need the following lemma, which is proved in Appendix~\ref{sec:independenceLemmaProof}.

\begin{lemma}[Independence lemma]  \label{lemma:independenceLemma}
	Consider a finite set $\Ac$ and a subset $\Ac' \subset \Ac$. Let $p_A$ be an arbitrary pmf over $\Ac$. Let the random vector $A^n$ be distributed proportionally to the product distribution $\prod_{l=1}^{n} p_A(a_l)$, restricted to the support set $\{a^n: a_k \in \Ac' \text{ for some $k$}\}$. Let $I$ be drawn uniformly from $\{i : A_i \in \Ac' \}$. Let $J = \left( (I+s-1) \mod n \right) + 1$ for some integer $s \in \natSet{(n-1)}$. Then, the random variables $A_I$ and $A_J$ are independent.
\end{lemma}

We bound the probability $\P( \Er_\text{d} \cond \Er_\text{e}^c )$ by the following lemma.
\begin{lemma} \label{lemma:decodingError}
	$\P(\Er_\text{d} \cond \Er_\text{e}^c) \to 0$ as $n\to\infty$ if
	\begin{align}
		R_3 &< H(Y\cond Z_1,Z_2,U) - \delta(\eps),  \label{eq:Per1}\\
		\tilde R_1 + R_3 &< H(Y\cond Z_2,U) + I(Z_1;Z_2\cond U) - \delta(\eps), \label{eq:Per2}\\
		\tilde R_2 + R_3 &< H(Y\cond Z_1,U) + I(Z_1;Z_2\cond U) - \delta(\eps), \label{eq:Per3}\\
		\tilde R_1 + \tilde R_2 + R_3&< H(Y\cond U) + I(Z_1;Z_2\cond U) - \delta(\eps), \label{eq:Per4}\\
		R_0 + \tilde R_1 + \tilde R_2 + R_3 &< H(Y) + I(Z_1;Z_2\cond U) - \delta(\eps). \label{eq:Per5}
	\end{align}  
\end{lemma}

\begin{proof}[Proof sketch]
The events of which $\Er_\text{d}$ is composed are  illustrated in Figure~\ref{fig:dcCodingScheme_marton}, which also depicts the structure of the codebook  for $m_0=1$. The product sets $\Lc_1(m_1)\times \Lc_2(m_2)$, for each $(m_{1},m_{2})$, are represented by shaded squares. In each product set, the sequence pair selected in step~\ref{item:MC_Gen_Error} of the codebook generation procedure is shown with its superposed $x^n$ codewords, as created in step~\ref{item:MC_Gen_superpose}. The correct codeword $x^n(1,1,1,1)$ is shown as a white circle which is connected to the received sequence $y^n$. The codewords that may be mistakenly detected at the receiver are shown as black circles. The product sets associated with decoding error events $\Er_{1}$, $\Er_{2}$, $\Er_{3}$, and $\Er_{4}$ are labeled $1$, $2$, $3$, and $4$, respectively.

\begin{figure}
	\centering
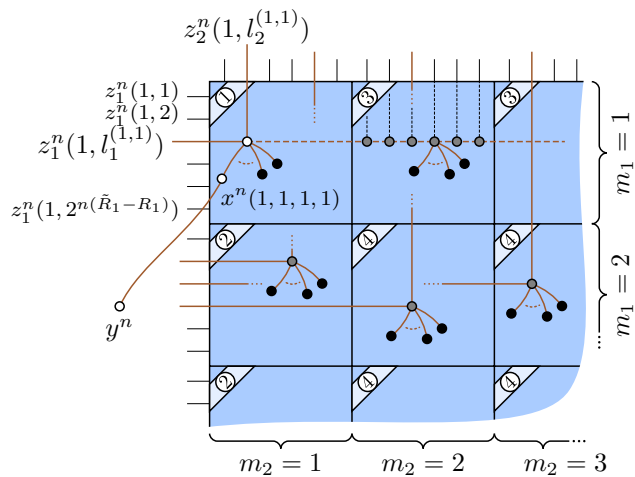%

	\caption{Illustration of decoding error events, for $m_0=1$.}
	\label{fig:dcCodingScheme_marton}
\end{figure}
We bound the probability of each sub-event of $\Er_\text{d}$. First, note that by the conditional typicality lemma in~\cite{ElGamalKim}, $\P(\Er_0^c) \to 0$ as $n\to \infty$ (this relies on $\eps'<\eps$). The probabilities of the events $\Er_1$ through $\Er_5$ conditioned on $\Er_\text{e}^c $ tend to zero as $n\to\infty$ under conditions~\eqref{eq:Per1} through~\eqref{eq:Per5}, correspondingly.

The events $\Er_2$ and $\Er_3$ require the most careful analysis, since the true codeword $x^n(1,1,1,1)$ and the codewords with which it may be confused can share the same $z_1^n$ or $z_2^n$ sequence (see dashed line and circles on it in Figure~\ref{fig:dcCodingScheme_marton}). 
Moreover, even when the chosen pairs in two different product sets do not share one of the two coordinates (see the chosen pairs for $(m_1,m_2)=(1,1)$ and $(2,1)$ in Figure~\ref{fig:dcCodingScheme_marton}), correlation could potentially be caused by the selection procedure in step~\ref{item:MC_Gen_Error} of codebook generation. We use the independence lemma (Lemma~\ref{lemma:independenceLemma}) to show that the event $\Er_\text{e}^c$ prevents this correlation leakage from occurring. The application of the lemma is what distinguishes this analysis from the conventional Marton inner bound for broadcast channels~\cite{Marton1979,ElGamalVanDerMeulen1981}. There, analysis of the selection process can be altogether avoided since each receiver decodes only one of the two coordinates.

A detailed proof for the event $\Er_3$ is given in Appendix~\ref{sec:decodingErrorProof}, the other events follow likewise.
\end{proof}

\vspace{2mm}
\noindent \emph{Analysis of disturbance rate.} When viewed by receiver $Z_1$, the codeword for message $m=(m_0,m_1,m_2,m_3)$
appears as $z_1^n(m_0, l_{1}^{(m_0,m_{1},m_{2})})$. We can pessimistically assume that all sequences $z_1^n(m_0,l_1)$ as created in step~\ref{item:MC_Gen_indep1} of codebook generation can be seen at the receiver for some message $m$.  Therefore, the number of possible sequences at $Z_1$, and thus its disturbance rate, is upper-bounded by $H(Z_1^n) \leq n(R_0 + \tilde R_1)$. 
Applying the same argument for $Z_2$, the proposed scheme achieves
\begin{align}
	R_0 + \tilde R_1 &\leq \Rdk 1, \label{eq:distRate1} \\ 
	R_0 + \tilde R_2 &\leq \Rdk 2. \label{eq:distRate2}
\end{align}

\vspace{2mm}
\noindent \emph{Conclusion of the proof.} Collecting inequalities~\eqref{eq:atLeastOne} through~\eqref{eq:distRate2}, recalling $R=R_0+R_1+R_2+R_3$, and using the Fourier-Motzkin procedure to eliminate $R_0$, $R_1$, $R_2$, and $R_3$ leads to the $(R,\Rdk1,\Rdk2)$ region claimed in the theorem.

Finally, the statement of Remark~\ref{rem:cond6_alt} follows from 
\begin{align*}
	& - I(Z_1;Z_2 \cond U) + I(Z_1;Z_2 \cond U,Y) \\
	& = - H( Z_2 \cond U) + H(Z_2 \cond U,Z_1) + H(Z_2 \cond U, Y) - H(Z_2 \cond U,Y,Z_1) \\
	& = - I(Y;Z_2 \cond U) + I(Y;Z_2\cond U,Z_1),
\end{align*}
which leads to the equality
\begin{align*}
	&  H(Y \cond Z_1,Z_2,U) + H(Y \cond U) - I(Z_1;Z_2 \cond U) + I(Z_1;Z_2 \cond U,Y) \\
	&= H(Y \cond Z_1,Z_2,U) + H(Y \cond U) - I(Y;Z_2 \cond U) + I(Y;Z_2\cond U,Z_1) \\
	&= H(Y \cond Z_1, U) +  H(Y \cond Z_2, U). 
\end{align*}

\subsection{Proof of Theorem~\ref{thm:2dc_OuterBound}} 
\label{sec:2dc_OuterBound_Proof}

First, consider 
\begin{align*}
	nR &\leq I(X^n;Y^n) + n\eps_n \\
	&= \sum_{i=1}^n I(X^n; Y_i \cond Y^{i-1}) + n\eps_n \\
	&= \sum_{i=1}^n I(X_i; Y_i \cond Y^{i-1}) + n\eps_n \\
	&= nI(X;Y\cond Q) \\
	&= nH(Y\cond Q).
\end{align*}
Furthermore,
\begin{align*}
	n \Rdk1 &\geq  I(X^n; Z_1^n) \\
	&\geq I(Y^n; Z_1^n) \\
	&= \sum_{i=1}^n I(Y_i; Z_1^n \cond Y^{i-1}) \\
	&\geq \sum_{i=1}^n I(Y_i; Z_{1i} \cond Y^{i-1}) \\
	&= nI(Y; Z_1 \cond Q),
\end{align*}
where $Y=Y_T$, $Z_1 = Z_{1T}$, and $Q = (Y^{T-1}, T)$ with $T \sim \mathrm{Unif}\natSet n$. The same argument leads to 
\begin{align*}
	n \Rdk2 &\geq nI(Y; Z_2 \cond Q),
\end{align*}
with the same random variable identifications, and the additional $Z_2 = Z_{2T}$.
Finally, the cardinality bound on $\Qc$ follows from the convex cover method in~\cite{ElGamalKim}.

\subsection{Proof of Theorem~\ref{thm:dcCapa_2_YZ1Z2}} \label{sec:proof_dcCapa_2_YZ1Z2}
First, we specialize Theorem~\ref{thm:dcCapa_2} as follows.
\begin{corollary}
\label{thm:dcCapa_2_roof} 
	The rate--disturbance region $\Rr$ of the deterministic channel with two disturbance constraints is inner-bounded by the set of rate triples $(R,\Rdk 1,\Rdk 2)$ such that
\begin{align}
		R & \leq H(Y), \label{eq:dcCapa2_roof1} \\
		\Rdk 1 &\geq I(Y; Z_1, U), \label{eq:dcCapa2_roof2} \\
		\Rdk 2 &\geq I(Y; Z_2, U), \label{eq:dcCapa2_roof3} \\
		\Rdk 1 + \Rdk 2  &\geq I(Y; Z_1, Z_2, U) + I(Y; U) + I(Z_1; Z_2\cond U)  \notag \\
		&= I(Y; Z_1, U) + I(Y; Z_2, U) + I(Z_1; Z_2\cond U, Y),  \label{eq:dcCapa2_roof4}
\end{align}
	for some pmf $p(u,x)$.
\end{corollary}

The two equivalent expressions in~\eqref{eq:dcCapa2_roof4} originate from Remark~\ref{rem:cond6_alt}. An example of the constituent regions of Corollary~\ref{thm:dcCapa_2_roof} for fixed $p(u,x)$ is depicted in Figure~\ref{fig:dcCapa2_constituentRegion_roof}. The figure also illustrates how the corollary follows from Theorem~\ref{thm:dcCapa_2}: Each constituent region of the corollary is a strict subset of the constituent region of the theorem, for the same $p(u,x)$.
\begin{figure}[htb]
	\centering
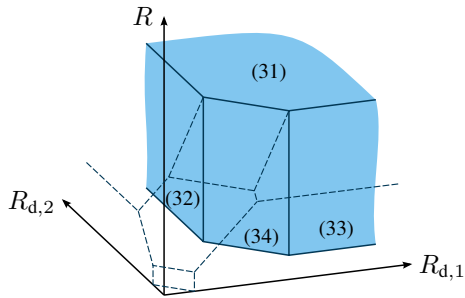%

	\caption{Constituent region for Corollary~\ref{thm:dcCapa_2_roof}, for a fixed $p(u,x)$. Each face is annotated by the inequality that defines it. For comparison, the constituent region of Theorem~\ref{thm:dcCapa_2} is shown with dashed lines (see Figure~\ref{fig:dcCapa2_constituentRegion}). }
	\label{fig:dcCapa2_constituentRegion_roof}
\end{figure}

\begin{proof}[Proof of Corollary~\ref{thm:dcCapa_2_roof}]
In Theorem~\ref{thm:dcCapa_2}, consider the case where~\eqref{eq:dcCapa2_cond1} is met with equality, i.e., $R = H(Y)$. This yields a subset region which is still achievable. It simplifies to
\begin{align}
		 \Rdk 1 + \Rdk 2 &\geq I(Z_1; Z_2\cond U), \label{eq:RHYdom1} \\
		 \Rdk 1 &\geq I(Y; Z_1, U), 
		 \\
		 \Rdk 2 &\geq I(Y; Z_2, U), 
		 \\
		 \Rdk 1 + \Rdk 2  &\geq I(Y; Z_1, Z_2, U) + I(Z_1; Z_2\cond U), \label{eq:RHYdom2} \\
		 \Rdk 1 + \Rdk 2  &\geq I(Y; Z_1, Z_2, U) + I(Y; U) + I(Z_1; Z_2\cond U) \notag 
		 \\
		 &= I(Y; Z_1, U) + I(Y; Z_2, U) + I(Z_1; Z_2\cond U, Y). \label{eq:RHYdom}
\end{align}
Clearly, conditions~\eqref{eq:RHYdom1} and~\eqref{eq:RHYdom2} are dominated by inequality~\eqref{eq:RHYdom}, and the desired result follows.
\end{proof}

\begin{proof}[Proof of achievability for Theorem~\ref{thm:dcCapa_2_YZ1Z2}] We further specialize Corollary~\ref{thm:dcCapa_2_roof}. 
We choose $U = Z_1 \vee Z_2$, i.e., the common part of $Z_1$ and $Z_2$. This implies that condition~\eqref{eq:dcCapa2_roof4} can be omitted, since $I(Z_1;Z_2 \cond U,Y)=0$ for all $p(u,x)$ by assumption. Furthermore, $U$ can be dropped from conditions~\eqref{eq:dcCapa2_roof2} and~\eqref{eq:dcCapa2_roof3} by virtue of being a function of $Z_1$ and $Z_2$. We conclude that
\begin{align}
		R & \leq H(Y), \\
		\Rdk 1 &\geq I(Y; Z_1), \\
		\Rdk 2 &\geq I(Y; Z_2), 
\end{align}
is achievable for all $p(x)$. Adding a time-sharing random variable $Q$ completes the proof.

Note that in the special case where $Y \refines Z_1$ or $Y \refines Z_2$, the same conclusion holds with the choice $U = \emptyset$.
\end{proof}

\subsection{Proof of Corollary~\ref{thm:dcCapa_2_degraded}} \label{sec:proof_dcCapa_2_degraded}
\begin{proof}[Proof of achievability]
We prove the result for $Z_1 \refines Z_2$, the other case follows by symmetry. 
We specialize the achievable region of Theorem~\ref{thm:dcCapa_2} by choosing $U = Z_2$.
The rate--disturbance constraints are
	\begin{align}
		 R  &\leq H(Y),  \\
		 \Rdk 1 + \Rdk 2 &\geq 0, \label{eq:Z1Z2degrpf_1} \\
		 R - \Rdk 1 &\leq H(Y\cond Z_1), \label{eq:Z1Z2degrpf_2} \\
		 R - \Rdk 2 &\leq H(Y\cond Z_2), \label{eq:Z1Z2degrpf_3} \\
		 R - \Rdk 1 - \Rdk 2  &\leq H(Y\cond Z_1), \label{eq:Z1Z2degrpf_4} \\
		 2R - \Rdk 1 - \Rdk 2  &\leq H(Y\cond Z_1)  + H(Y\cond Z_2) \label{eq:Z1Z2degrpf_5}.
	\end{align}
Clearly,~\eqref{eq:Z1Z2degrpf_1} is vacuous. Furthermore,~\eqref{eq:Z1Z2degrpf_4} is dominated by~\eqref{eq:Z1Z2degrpf_2}, and~\eqref{eq:Z1Z2degrpf_5} is dominated by the sum of~\eqref{eq:Z1Z2degrpf_2} and~\eqref{eq:Z1Z2degrpf_3}. This completes the proof.
\end{proof}

\begin{proof}[Proof of converse] The first inequality follows from Fano's inequality as 
	\begin{align*}
		nR &\leq I(X^n;Y^n) + n\eps_n \\
		&= H(Y^n) + n\eps_n \\
		&\leq nH(Y) + n\eps_n,
	\end{align*}
where $Y = Y_Q$ and $Q \sim \mathrm{Unif}\natSet n$. The other two inequalities follow as 
	\begin{align*}
		n(R - \Rdk1)
		&\leq nR - I(X^n;Z_1^n) \\
		&\leq H(Y^n) - H(Z_1^n) + n\eps_n \\
		&\leq  H(Y^n, Z_1^n) - H(Z_1^n) + n\eps_n \\
		&= H(Y^n\cond Z_1^n) + n\eps_n \\
		&\leq nH(Y\cond Z_1) + n\eps_n,
	\end{align*}
with $Z_1 = Z_{1Q}$, and likewise for $n(R - \Rdk2)$. 
\end{proof}

\section{Acknowledgments}
The authors would like to thank Pramod Viswanath and Yeow-Khiang Chia for helpful discussions, and gratefully acknowledge the Information Theoretic Inequalities Prover (Xitip)~\cite{Xitip2008,Itip}, which was used as a verification tool in some of our derivations.

\bibliographystyle{IEEEtran}
\bibliography{IEEEabrv,../unified-bibtex/references-unified,../unified-bibtex/myPublications}

\appendix

\subsection{Proof of Lemma~\ref{lemma:encodingError}}  \label{sec:encodingErrorProof}
\newcommand{\pPair}{p}
\newcommand{\nRow}{m}  
\newcommand{\nCol}{l}   
\newcommand{\kRow}{K}
The product bin $(m_{1},m_{2})=(1,1)$ for $m_0=1$ contains $\nCol \nRow$ sequence pairs, where  $\nCol = 2^{nr_{1}}$ and $\nRow = 2^{nr_{2}}$. Each pair $(Z^n_{1}(1,l_{1}), Z^n_{2}(1,l_{2}))$, for $l_{1} \in \natSet{\nCol}$ and $l_{2} \in \natSet{\nRow}$, has probability $\pPair \doteq 2^{-nI(Z_{1}; Z_{2}\cond U)}$ to be jointly typical. Now fix one coordinate, say $l_{1}=1$. The corresponding ``row'' of the bin contains $\nRow$ sequences $Z^n_2(1,l_2)$, each of which has an independent probability of $\pPair$ to be jointly typical with $Z^n_{1}(1,1)$. Let $\kRow$ be the total number of typical sequences in this row. Then 
\begin{align*}
	\P( \kRow = 0 ) & = (1-\pPair)^{\nRow}, \\
	\P( \kRow = 1 ) & = \nRow \pPair (1-\pPair)^{\nRow-1}, \\
	\P( \kRow \geq 2 ) & = 1 - (1-\pPair+\nRow \pPair)\underbrace{(1-\pPair)^{\nRow-1}}_{\geq 1-(\nRow-1)\pPair}  \\
	&\leq \nRow^2 \pPair^2.
\end{align*}
We have thus upper-bounded the probability to encounter two or more typical pairs in a single row. Consequently, the probability of two or more typical pairs occurring in \emph{any} row is upper bounded by $\nCol \nRow^2 \pPair^2$. Substituting definitions leads to the desired inequality. The same argument can be made for columns of the bin.

\subsection{Proof of independence lemma (Lemma~\ref{lemma:independenceLemma})} \label{sec:independenceLemmaProof}

	We prove the lemma for $s=1$, the remaining cases follow by symmetry. For ease of notation, define the specialized modulo operator $\modn{x} = 1+((x-1)\!\mod n)$, the indicator function $\ind_{\Ac'}(a)=1$ if $a\in \Ac'$ and 0 otherwise, and the shorthand notations $Y=A_I$ and $Z=A_J$. Notice that
	\begin{align*}
		p( a^n ) &= \begin{cases}
			\frac 1 c \prod_{l=1}^{n} p_A(a_l) & \text{if $a_k \in \Ac'$ for some $k \in \natSet n$} \\
			0 & \text{otherwise},
		\end{cases}
	\end{align*}
	where $c$ is a normalization constant, the exact value of which is not relevant. Further,
	\begin{align*}
		p(i \cond a^n ) &= \begin{cases}
			\frac 1 {\sum_{k=1}^{n} \ind_{\Ac'}(a_k)}  & \text{if $a_i \in \Ac'$} \\
			0 & \text{otherwise}.
		\end{cases}
	\end{align*}
	The joint distribution of $(A^n,I,J,Y,Z)$ is then
	\begin{align*}
		p(a^n, i, j, y, z) &= \begin{cases}
			\frac {p(a^n)}{\sum_{k=1}^{n} \ind_{\Ac'}(a_k)}  & \text{if $a_i \in \Ac'$, $a_i=y$, $a_j = z$, and $j = \modn{i+1}$} \\
			0 & \text{otherwise}.
		\end{cases}
	\end{align*}
	Partially marginalizing, it follows that
	\begin{align*}
		p(y,z) &= \sum_{i=1}^{n} \sum_{\substack{
			a^n: \ a_i \in \Ac' \\
			a_i = y \\
			a_{\modn{i+1}} = z
		}} 
		\frac {p(a^n)}{\sum_{k=1}^{n} \ind_{\Ac'}(a_k)}.
	\end{align*}
	It is clear that $p(y,z)=p(y)p(z)= 0$ if $y \notin \Ac'$. On the other hand, for $y \in \Ac'$, we have
	\begin{align*}
		p(y,z) 
		&= \sum_{i=1}^{n} \sum_{\substack{
			a^n:\  a_i = y \\
			a_{\modn{i+1}} = z
		}} 
		\frac {\prod_{l=1}^{n} p_A(a_l)}{c\, \sum_{k=1}^{n} \ind_{\Ac'}(a_k)}.
	\end{align*}
	The fraction under the sum is invariant under permutations of $a^n$. Therefore,
	\begin{align*}
		p(y,z) 
		&= \frac 1 c \sum_{i=1}^{n} \sum_{\substack{
			a^n:\  a_1 = y \\
			a_2 = z
		}} 
		\frac {\prod_{l=1}^{n} p_A(a_l)}{\sum_{k=1}^{n} \ind_{\Ac'}(a_k)} \\
		&= \frac n c \sum_{a^n = (y,z,a_3^n)} 
		\frac {\prod_{l=1}^{n} p_A(a_l)}{\sum_{k=1}^{n} \ind_{\Ac'}(a_k)} \\
		&= \frac {n \ p_A(y)\  p_A(z)} c \sum_{a_3^n \in \Ac^{n-2}} 
		\frac {\prod_{l=3}^{n} p_A(a_l)}{1 + \ind_{\Ac'}(z) + \sum_{k=3}^{n} \ind_{\Ac'}(a_k)},
	\end{align*}
	where $a_3^n$ are the last $n-2$ components of $a^n$. Observe that $p(y,z)$ separates into a function of $z$ and a function of $y$. Independence is thus established.

\subsection{Proof of Lemma~\ref{lemma:decodingError}, exemplified for $\Er_3$} \label{sec:decodingErrorProof}
We analyze the probability of $\Er_3$ as follows.
\begin{align*}
	\Er_{3} &= 
	\bigl \{ \bigl ( U^n(1), 
	Z^n_{1}(1, L_{1}^{(1,1, m_{2} )}), 
	Z^n_{2}(1, L_{2}^{(1,1, m_{2} )}),  \\
	& \qquad X^n(1,L_{1}^{(1,1, m_{2} )}, L_{2}^{(1,1, m_{2} )}, m_3),  
	Y^n \bigr) \in \Typ,  \\
	& \qquad \text{for some $m_{2}\neq 1$, $m_3$} \bigr\} \\
	&\subseteq 
	\bigl\{ \bigl( U^n(1), 
	Z^n_{1}(1, L_{1}^{(1,1, m_{2} )}), 
	Z^n_{2}(1, l_2), \\  
	& \qquad X^n(1, L_{1}^{(1, 1, m_{2})}, l_2, m_3),  
	Y^n \bigr) \in \Typ, \\
	&\qquad  \text{for some $m_{2}\neq 1$, $m_3$, $l_2 \notin \Lc_2(1)$} \bigr\} ,
\end{align*}
Define the event $\Es = \{L_{1}^{(1,1, m_{2} )} = L_{1}^{(1,1, 1 )} \}$, which allows us to write $\P(\Er_3 \cond \Er_\text{e}^c) = \P(\Er_3 \cap \Es \cond \Er_\text{e}^c) + \P(\Er_3 \cap \Es^c \cond \Er_\text{e}^c) $. We consider both terms separately.
\begin{align*}
	\Er_{3} \cap \Es & \subseteq
	\bigl \{ \bigl( U^n(1), 
	Z^n_{1}(1,L_{1}^{(1,1,1 )}), 
	Z^n_{2}(1,l_2), \\  
	&\qquad X^n(1,L_{1}^{(1,1, 1)}, l_2, m_3),  
	Y^n \bigr) \in \Typ,  \\
	&\qquad \text{for some $l_2 \notin \Lc_2(1)$, $m_3$} \bigr\}.
\end{align*}
Thus,
\begin{align*}
	& \P(\Er_{3} \cap \Es \cond \Er_\text{e}^c) \\
	& \leq \hspace{-5mm} \sum_{(u^n,z_{1}^n,y^n) \in \Typ} \hspace{-8mm} 
		\P\left( U^n(1)=u^n, Z^n_{1}(1,L_{12}^{(1, 1, 1)})=z_1^n, Y^n=y^n \cond \Er_\text{e}^c \right) \\
	& \qquad \cdot 
		\sum_{l_{2}\notin \Lc_2(1)} \sum_{m_{3}=1}^{2^{nR_3}}  
		\P\bigl( 
		(u^n,z_{1}^n,Z^n_{2}(1,l_{2}), \\[-4mm]
	& \hspace{33mm} 
		X^n(1,L_{1}^{(1,1,1)}, l_{2}, m_3),y^n) \in \Typ \cond \Er_\text{e}^c \bigr) \\
	& \leq 2^{n(\tilde R_{2} + R_{3} )} \ P^\star,
\end{align*}
where $P^\star$ is shorthand for the last $\P(\cdot)$ expression. Continue with
{\allowdisplaybreaks
\begin{align*}
	P^\star 
	& = 
	\hspace{-8mm} \sum_{\substack{(z_{2}^n,x^n) \in  \Typ(\\ Z_{2},  X\cond u^n,z_{1}^n,y^n)}} \hspace{-6mm} 
		\P \left( Z_2^n(1,l_2)= z_2^n, X^n(1,L_1^{(1,1,1)},l_2,m_3)=x^n \right\vert \\*[-6mm]
	& \hspace{23mm} \left. U^n(1)=u^n, Z_1^n(L_1^{(1,1,1)})=z_1^n, Y^n=y^n, \Er_\text{e}^c \right) \\
	&\anneq{a}  \underbrace
			{\sum_{\substack{(z_{2}^n,x^n) \in  \Typ(\\ Z_{2},  X\cond u^n,z_{1}^n,y^n)}}}_
			{\mathop{\doteq} 2^{nH(X,Z_2|Z_1,Y,U)}}
		\underbrace{
			p(z_{2}^n\cond u^n)}_
			{\mathop{\doteq} 2^{-nH(Z_2|U)}} \quad
		\underbrace{
			p(x^n\cond z_1^n, z_2^n, u^n)}_
			{\mathop{\doteq} 2^{-nH(X|Z_1,Z_2,U)}} \\
	& \leq 2^{n( H(X,Z_2 | Z_1,Y,U) - H(Z_{2}|U) - H(X|Z_{1},Z_{2},U)  + \delta(\eps) )} \\
	& = 2^{n( - H(Y|Z_{1},U) - I(Z_{1}; Z_{2} | U) + \delta(\eps) )} .
\end{align*}
}%
In step (a), we have used the fact that $l_2 \notin \Lc_2(1)$, and therefore, $Z_2^n(1,l_2)$ relates to a bin other than the first one. It is independent of the conditions $Y^n=y^n$ and $\Er_\text{e}^c$, both of which relate only to the $(1,1)$ bin for $m_0=1$. A similar argument applies to the second term.

Substituting back in the previous chain of inequalities implies that $\P(\Er_{3} \cap \Es \cond \Er_\text{e}^c) \to 0$ as  $n \to \infty$ if inequality~\eqref{eq:Per3} holds.

Next, consider 
\begin{align*}
	\Er_{3} \cap \Es^c  
	&\subseteq \bigl\{ \bigl( U^n(1), 
	Z^n_{1}(1, l_1), 
	Z^n_{2}(1, l_2),  
	X^n(1, l_1, l_2, m_3), \\  
	& \hspace{8.5mm} Y^n \bigr) \in \Typ,  \text{ for some $l_1 \in \Lc_1(1) \setminus \{ L_1^{(1,1,1)} \}$,} \\
	& \hspace{42mm} \text{ $l_2 \notin \Lc_2(1)$, $m_3$ } \bigr \}  .
\end{align*}
We argue
\begin{align*}
	& \P(\Er_{3} \cap \Es^c \cond \Er_\text{e}^c) \\
	& \leq \sum_{(u^n,y^n) \in \Typ} 
		\P\left( U^n(1)=u^n, Y^n=y^n \cond \Er_\text{e}^c \right) 
		\hspace{-8mm} \sum_{l_{1}\in \Lc_1(1) \setminus \{ L_1^{(1,1,1)} \} } \hspace{-8mm} \cdots \\
	& \qquad \cdots   
		\sum_{l_{2}\notin \Lc_2(1)} 
		\sum_{m_{3}=1}^{2^{nR_3}}  
		\P\bigl( (u^n,Z_{1}^n(1,l_1),Z^n_{2}(1,l_{2}), \\[-4mm]
		& \hspace{36.5mm} X^n(1,l_1, l_{2}, m_3),y^n) \in \Typ \cond U^n(1)=u^n, Y^n=y^n, \Er_\text{e}^c \bigr) \\
	& \leq 2^{n(\tilde R_1 - R_1  + \tilde R_{2} + R_{3} )} \ P^\star,
\end{align*}
where $P^\star$ represents the last $\P(\cdot)$ expression. Finally,
\begin{align*}
	& P^\star = \hspace{-5mm}
	\sum_{\substack{(z_1^n,z_2^n,x^n) \in  \Typ(\\ Z_1, Z_{2},  X\cond u^n,y^n)}} 
		\P \bigl( Z_1^n(1,l_1)=z_1^n, Z_2^n(1,l_2)= z_2^n, \\[-7mm]
	& \hspace{40mm} X^n(1,l_1,l_2,m_3)=x^n \cond \\
	& \hspace{40mm} U^n(1)=u^n,  Y^n=y^n, \Er_\text{e}^c \bigr) \\
	& = 
	\sum_{\substack{(z_1^n,z_2^n,x^n) \in  \Typ(\\ Z_1, Z_2,  X\cond u^n,y^n)}} \ \  
	\sum_{\substack{z_2^n(l'_2),\text{ for }\\ \text{all } l'_2 \in \Lc_2(1)}}
		\P\bigl( Z_2^n(1,l'_2) = z_2^n(l'_2) \text{ for} \\[-8mm]
		& \hspace{56mm} \text{all } l'_2 \in \Lc_2(1) \cond \Er_\text{e}^c \bigr) \\[3mm]
		&\qquad \qquad \cdot \P \bigl( Z_1^n(1,l_1)=z_1^n, Z_2^n(1,l_2)= z_2^n, \\
		& \hspace{20mm} X^n(1,l_1,l_2,m_3)=x^n \cond \\
		& \hspace{20mm} U^n(1)=u^n,  Y^n=y^n, Z_2^n(1,l'_2) = z_2^n(l'_2) \\
		& \hspace{20mm} \text{for all } l'_2 \in \Lc_2(1), \Er_\text{e}^c \bigr) \\
	& \annleq{a} \hspace{-5mm} \underbrace
			{\sum_{\substack{(z_1^n,z_2^n,x^n) \in  \Typ(\\ Z_1, Z_2, X\cond u^n,y^n)}}}_
			{\mathop{\doteq} 2^{nH(X,Z_1,Z_2|Y,U)}}
		\hspace{0mm} \underbrace{
			p(z_{1}^n\cond u^n, \Er_\text{e}^c)}_
			{\substack{\text{(b)} \\[1mm] \mathop{\doteq} 2^{-nH(Z_1|U)}}} 
		\underbrace{
			p(z_{2}^n\cond u^n)}_
			{\mathop{\doteq} 2^{-nH(Z_2|U)}} 
		\underbrace{
			p(x^n\cond z_1^n, z_2^n, u^n)}_
			{\mathop{\doteq} 2^{-nH(X|Z_1,Z_2,U)}} \\
	& \leq 2^{n( H(X,Z_1,Z_2 | Y,U) - H(Z_{1}|U) - H(Z_{2}|U) - H(X|Z_{1},Z_{2},U)  + \delta(\eps) )} \\
	& = 2^{n( - H(Y|U) - I(Z_{1}; Z_{2} | U) + \delta(\eps) )} .
\end{align*}
Here, (a) uses uses the fact that for the $l_1$ indices in question, $Z_1^n(1,l_1)$ is independent of $Y^n$. This is a consequence of independence between the selected $Z_1^n(1,L_1^{(1,1,1)})$ and the other (non-selected) $Z_1^n(1,l_1)$ due to Lemma~\ref{lemma:independenceLemma}. The lemma applies because the event is conditioned  (1)  on $\Er_\text{e}^c$, which ensures that picking $L_1^{(1,1,1)}$ is uniform as required by the lemma, and (2)  on $Z_2^n(1,l'_2)$ for all $l'_2  \in \Lc_2(1)$, which provides for the qualifying set $\Ac'$ of the lemma.

Step (b) follows from 
\begin{align*}
	p(z_{1}^n \cond u^n, \Er_\text{e}^c) &= p(z_{1}^n \cond u^n) \cdot \frac { p(\Er_\text{e}^c\cond u^n, z_{1}^n ) }{ p(\Er_\text{e}^c\cond u^n) } \\
	&\leq p(z_{1}^n \cond u^n) \cdot \frac { 1 }{ p(\Er_\text{e}^c\cond u^n) } \\
	&\leq p(z_{1}^n \cond u^n) \cdot \frac { 1 }{ 1- 2^{-\delta n} } \\
	&\leq 2^{-n(H(Z_1|U)-\eps)} \cdot 2^{n \delta'} \\
	&\leq 2^{-n(H(Z_1|U)-\eps-\delta')} .
\end{align*}
Here, $\delta$ is the minimum slack of the three conditions for $\Er_\text{e}^c$ in Lemma~\ref{lemma:encodingError}. Note that for any $\delta, \delta' >0$, we can find an $N_0$ such that
\begin{align*}
	\forall n\geq N_0: \frac{1}{1-2^{-\delta n}} \leq 2^{n \delta'}.
\end{align*}
We conclude that $\P(\Er_{3} \cap \Es^c \cond \Er_\text{e}^c) \to 0$ as $n \to \infty$ if 
\begin{align*}
	\tilde R_{1} - R_{1} + \tilde R_{2} + R_3 & \leq H(Y|Q) +  I(X_{1}; X_{2} | Q) - \delta(\eps).
\end{align*}
This is an implication of~\eqref{eq:Per4} which stems from analyzing $\Er_4$, and may thus be omitted.

\end{document}

%% file: defns.tex
\usepackage{xspace}
\usepackage{mathrsfs}
\usepackage{amsopn}

\def\tr{\mathop{\rm tr}\nolimits}%
\newcommand{\Ac}{\mathcal{A}}

\newcommand{\Cc}{\mathcal{C}}

\newcommand{\Mc}{\mathcal{M}}

\newcommand{\Qc}{\mathcal{Q}}

\newcommand{\Sc}{\mathcal{S}}

\newcommand{\Uc}{\mathcal{U}}

\newcommand{\Xc}{\mathcal{X}}
\newcommand{\Yc}{\mathcal{Y}}
\newcommand{\Zc}{\mathcal{Z}}

\newcommand{\Rr}{\mathscr{R}}

\newcommand{\pen}{{P_e^{(n)}}}

\def\eps{\epsilon}

\let\P\relax
\DeclareMathOperator\P{P}

\DeclareMathOperator\C{C}

\newcommand{\N}{\mathcal{N}}

\def\textiid{i.i.d.\@\xspace}
\newcommand\iid{\ifmmode\text{ i.i.d. } \else \textiid \fi}

\newcommand{\Real}{\mathbb{R}}

\newcommand{\ind}{\boldsymbol{1}}


%% file: includesvgNoBuild.tex
\usepackage{calc}

\newcommand{\fn}{\footnotesize}  
\newcommand{\scrs}{\scriptsize}

%% file: figs/f_disturbanceChannelK_oneBox.pdf_tex
\begingroup%
  \makeatletter%
  \providecommand\color[2][]{%
    \errmessage{(Inkscape) Color is used for the text in Inkscape, but the package 'color.sty' is not loaded}%
    \renewcommand\color[2][]{}%
  }%
  \providecommand\transparent[1]{%
    \errmessage{(Inkscape) Transparency is used (non-zero) for the text in Inkscape, but the package 'transparent.sty' is not loaded}%
    \renewcommand\transparent[1]{}%
  }%
  \providecommand\rotatebox[2]{#2}%
  \ifx\svgwidth\undefined%
    \setlength{\unitlength}{249.44882813bp}%
    \ifx\svgscale\undefined%
      \relax%
    \else%
      \setlength{\unitlength}{\unitlength * \real{\svgscale}}%
    \fi%
  \else%
    \setlength{\unitlength}{\svgwidth}%
  \fi%
  \global\let\svgwidth\undefined%
  \global\let\svgscale\undefined%
  \makeatother%
  \begin{picture}(1,0.35458174)%
    \put(0,0){\includegraphics[width=\unitlength]{f_disturbanceChannelK_oneBox.pdf}}%
    \put(0.04545455,0.28523682){\color[rgb]{0,0,0}\makebox(0,0)[rb]{\smash{$M$}}}%
    \put(0.31249999,0.30796407){\color[rgb]{0,0,0}\makebox(0,0)[b]{\smash{$X^n$}}}%
    \put(0.55113636,0.30796407){\color[rgb]{0,0,0}\makebox(0,0)[b]{\smash{$Y^n$}}}%
    \put(0.55681817,0.18296407){\color[rgb]{0,0,0}\makebox(0,0)[lb]{\smash{$Z_1^n$}}}%
    \put(0.55681817,0.0466004){\color[rgb]{0,0,0}\makebox(0,0)[lb]{\smash{$Z_K^n$}}}%
    \put(0.64772724,0.18296407){\color[rgb]{0,0,0}\makebox(0,0)[lb]{\smash{$\tfrac 1 n I(X^n;Z_1^n) \leq \Rdk 1$}}}%
    \put(0.64772724,0.0466004){\color[rgb]{0,0,0}\makebox(0,0)[lb]{\smash{$\tfrac 1 n I(X^n;Z_K^n) \leq \Rdk K$}}}%
    \put(0.82954538,0.28523682){\color[rgb]{0,0,0}\makebox(0,0)[lb]{\smash{$\hat M$}}}%
    \put(0.18749999,0.28523682){\color[rgb]{0,0,0}\makebox(0,0)[b]{\smash{\small Encoder}}}%
    \put(0.4431818,0.17728612){\color[rgb]{0,0,0}\rotatebox{90}{\makebox(0,0)[b]{\smash{$p(y,z_1,\dots,z_K|x)$}}}}%
    \put(0.67613636,0.28523682){\color[rgb]{0,0,0}\makebox(0,0)[b]{\smash{\small Decoder}}}%
  \end{picture}%
\endgroup%

%% file: figs/f_dcCapa_constituentRegion.pdf_tex

\begingroup
  \makeatletter
  \providecommand\color[2][]{%
    \errmessage{(Inkscape) Color is used for the text in Inkscape, but the package 'color.sty' is not loaded}
    \renewcommand\color[2][]{}%
  }
  \providecommand\transparent[1]{%
    \errmessage{(Inkscape) Transparency is used (non-zero) for the text in Inkscape, but the package 'transparent.sty' is not loaded}
    \renewcommand\transparent[1]{}%
  }
  \providecommand\rotatebox[2]{#2}
  \ifx\svgwidth\undefined
    \setlength{\unitlength}{249.44882813pt}
  \else
    \setlength{\unitlength}{\svgwidth}
  \fi
  \global\let\svgwidth\undefined
  \makeatother
  \begin{picture}(1,0.52741502)%
    \put(0,0){\includegraphics[width=\unitlength]{f_dcCapa_constituentRegion.pdf}}%
    \put(0.62784881,0.36764521){\color[rgb]{0,0,0}\makebox(0,0)[lb]{\smash{$A$}}}%
    \put(0.46875793,0.1971907){\color[rgb]{0,0,0}\makebox(0,0)[lb]{\smash{$B$}}}%
    \put(0.54830337,0.27673609){\color[rgb]{0,0,0}\makebox(0,0)[lb]{\smash{\small 45\degree}}}%
    \put(0.684667,0.03809972){\color[rgb]{0,0,0}\makebox(0,0)[lb]{\smash{$R$}}}%
    \put(0.32103067,0.51537247){\color[rgb]{0,0,0}\makebox(0,0)[rb]{\smash{$\Rd$}}}%
    \put(0.46307611,0.35628158){\color[rgb]{0,0,0}\makebox(0,0)[b]{\smash{$\Rr(U,X)$}}}%
    \put(0.62784881,0.00400884){\color[rgb]{0,0,0}\makebox(0,0)[b]{\smash{\small $I(X;Y)$}}}%
    \put(0.44603067,0.00400884){\color[rgb]{0,0,0}\makebox(0,0)[b]{\smash{\small $I(X;Y|U)$}}}%
    \put(0.30966705,0.20855433){\color[rgb]{0,0,0}\makebox(0,0)[rb]{\smash{\small $I(X;Z|U)$}}}%
  \end{picture}%
\endgroup

%% file: figs/f_injectiveDet2ic_connection2.pdf_tex

\begingroup
  \makeatletter
  \providecommand\color[2][]{%
    \errmessage{(Inkscape) Color is used for the text in Inkscape, but the package 'color.sty' is not loaded}
    \renewcommand\color[2][]{}%
  }
  \providecommand\transparent[1]{%
    \errmessage{(Inkscape) Transparency is used (non-zero) for the text in Inkscape, but the package 'transparent.sty' is not loaded}
    \renewcommand\transparent[1]{}%
  }
  \providecommand\rotatebox[2]{#2}
  \ifx\svgwidth\undefined
    \setlength{\unitlength}{249.44882813pt}
  \else
    \setlength{\unitlength}{\svgwidth}
  \fi
  \global\let\svgwidth\undefined
  \makeatother
  \begin{picture}(1,0.32042373)%
    \put(0,0){\includegraphics[width=\unitlength]{f_injectiveDet2ic_connection2.pdf}}%
    \put(0.53409091,0.29819841){\color[rgb]{0,0,0}\makebox(0,0)[b]{\smash{$Y'_1$}}}%
    \put(0.53409091,0.04819846){\color[rgb]{0,0,0}\makebox(0,0)[b]{\smash{$Y'_2$}}}%
    \put(0.75000002,0.25114032){\color[rgb]{0,0,0}\makebox(0,0)[lb]{\smash{$Y_1 \to \hat M_1$}}}%
    \put(0.75000002,0.04659493){\color[rgb]{0,0,0}\makebox(0,0)[lb]{\smash{$Y_2 \to \hat M_2$}}}%
    \put(0.25000002,0.27386762){\color[rgb]{0,0,0}\makebox(0,0)[rb]{\smash{$M_1 \to X_1$}}}%
    \put(0.25000002,0.02386757){\color[rgb]{0,0,0}\makebox(0,0)[rb]{\smash{$M_2 \to X_2$}}}%
    \put(0.55681821,0.11477669){\color[rgb]{0,0,0}\makebox(0,0)[lb]{\smash{$Z_1$}}}%
    \put(0.55681821,0.1829585){\color[rgb]{0,0,0}\makebox(0,0)[lb]{\smash{$Z_2$}}}%
    \put(0.38636365,0.27954944){\color[rgb]{0,0,0}\makebox(0,0)[b]{\smash{$g_{11}$}}}%
    \put(0.38636365,0.21136762){\color[rgb]{0,0,0}\makebox(0,0)[b]{\smash{$g_{12}$}}}%
    \put(0.38636365,0.0977313){\color[rgb]{0,0,0}\makebox(0,0)[b]{\smash{$g_{21}$}}}%
    \put(0.38636365,0.02954944){\color[rgb]{0,0,0}\makebox(0,0)[b]{\smash{$g_{22}$}}}%
    \put(0.65340909,0.25114032){\color[rgb]{0,0,0}\makebox(0,0)[b]{\smash{$f_1$}}}%
    \put(0.65340909,0.04659493){\color[rgb]{0,0,0}\makebox(0,0)[b]{\smash{$f_2$}}}%
  \end{picture}%
\endgroup

%% file: figs/f_exampleDisturbanceChannel1.pdf_tex

\begingroup
  \makeatletter
  \providecommand\color[2][]{%
    \errmessage{(Inkscape) Color is used for the text in Inkscape, but the package 'color.sty' is not loaded}
    \renewcommand\color[2][]{}%
  }
  \providecommand\transparent[1]{%
    \errmessage{(Inkscape) Transparency is used (non-zero) for the text in Inkscape, but the package 'transparent.sty' is not loaded}
    \renewcommand\transparent[1]{}%
  }
  \providecommand\rotatebox[2]{#2}
  \ifx\svgwidth\undefined
    \setlength{\unitlength}{249.44882813pt}
  \else
    \setlength{\unitlength}{\svgwidth}
  \fi
  \global\let\svgwidth\undefined
  \makeatother
  \begin{picture}(1,0.30912879)%
    \put(0,0){\includegraphics[width=\unitlength]{f_exampleDisturbanceChannel1.pdf}}%
    \put(0.1931818,0.05800434){\color[rgb]{0,0,0}\makebox(0,0)[rb]{\smash{$M \to X$}}}%
    \put(0.22727271,0.08073159){\color[rgb]{0,0,0}\makebox(0,0)[lb]{\smash{\fn $\{0,1,2,3\}$}}}%
    \put(0.48863634,0.10118637){\color[rgb]{0,0,0}\makebox(0,0)[lb]{\smash{\fn $0$}}}%
    \put(0.48863634,0.0727773){\color[rgb]{0,0,0}\makebox(0,0)[lb]{\smash{\fn $1$}}}%
    \put(0.48863634,0.04436818){\color[rgb]{0,0,0}\makebox(0,0)[lb]{\smash{\fn $2$}}}%
    \put(0.5795454,0.10118637){\color[rgb]{0,0,0}\makebox(0,0)[rb]{\smash{\fn $0$}}}%
    \put(0.5795454,0.01595916){\color[rgb]{0,0,0}\makebox(0,0)[rb]{\smash{\fn $2$}}}%
    \put(0.48863634,0.01595916){\color[rgb]{0,0,0}\makebox(0,0)[lb]{\smash{\fn $3$}}}%
    \put(0.5795454,0.05914123){\color[rgb]{0,0,0}\makebox(0,0)[rb]{\smash{\fn $1$}}}%
    \put(0.48863631,0.27164092){\color[rgb]{0,0,0}\makebox(0,0)[lb]{\smash{\fn $0$}}}%
    \put(0.48863631,0.24323185){\color[rgb]{0,0,0}\makebox(0,0)[lb]{\smash{\fn $1$}}}%
    \put(0.48863631,0.21482273){\color[rgb]{0,0,0}\makebox(0,0)[lb]{\smash{\fn $2$}}}%
    \put(0.57954538,0.2580048){\color[rgb]{0,0,0}\makebox(0,0)[rb]{\smash{\fn $0$}}}%
    \put(0.48863631,0.18641371){\color[rgb]{0,0,0}\makebox(0,0)[lb]{\smash{\fn $3$}}}%
    \put(0.57954538,0.20118661){\color[rgb]{0,0,0}\makebox(0,0)[rb]{\smash{\fn $1$}}}%
    \put(0.74999994,0.2284589){\color[rgb]{0,0,0}\makebox(0,0)[lb]{\smash{$Z$}}}%
    \put(0.61363636,0.25118615){\color[rgb]{0,0,0}\makebox(0,0)[lb]{\smash{\fn $\{0,1\}$}}}%
    \put(0.80681817,0.05800434){\color[rgb]{0,0,0}\makebox(0,0)[lb]{\smash{$Y \to \hat M$}}}%
    \put(0.61363636,0.08073159){\color[rgb]{0,0,0}\makebox(0,0)[lb]{\smash{\fn $\{0,1,2\}$}}}%
  \end{picture}%
\endgroup

%% file: figs/f_exampleDisturbanceChannel1_region.pdf_tex

\begingroup
  \makeatletter
  \providecommand\color[2][]{%
    \errmessage{(Inkscape) Color is used for the text in Inkscape, but the package 'color.sty' is not loaded}
    \renewcommand\color[2][]{}%
  }
  \providecommand\transparent[1]{%
    \errmessage{(Inkscape) Transparency is used (non-zero) for the text in Inkscape, but the package 'transparent.sty' is not loaded}
    \renewcommand\transparent[1]{}%
  }
  \providecommand\rotatebox[2]{#2}
  \ifx\svgwidth\undefined
    \setlength{\unitlength}{249.44882813pt}
  \else
    \setlength{\unitlength}{\svgwidth}
  \fi
  \global\let\svgwidth\undefined
  \makeatother
  \begin{picture}(1,0.55306166)%
    \put(0,0){\includegraphics[width=\unitlength]{f_exampleDisturbanceChannel1_region.pdf}}%
    \put(0.06818185,0.52989672){\color[rgb]{0,0,0}\makebox(0,0)[rb]{\smash{$\Rd$}}}%
    \put(0.89772724,0.00148765){\color[rgb]{0,0,0}\makebox(0,0)[lb]{\smash{$R$}}}%
    \put(0.17045461,0.08671493){\color[rgb]{0,0,0}\makebox(0,0)[lb]{\smash{\scrs Single-user codebooks}}}%
    \put(0.17045461,0.12080581){\color[rgb]{0,0,0}\makebox(0,0)[lb]{\smash{\scrs Superposition codebooks}}}%
    \put(0.17442125,0.00148765){\color[rgb]{0,0,0}\makebox(0,0)[b]{\smash{\scrs 0.2}}}%
    \put(0.25793335,0.00148765){\color[rgb]{0,0,0}\makebox(0,0)[b]{\smash{\scrs 0.4}}}%
    \put(0.34144549,0.00148765){\color[rgb]{0,0,0}\makebox(0,0)[b]{\smash{\scrs 0.6}}}%
    \put(0.42495759,0.00148765){\color[rgb]{0,0,0}\makebox(0,0)[b]{\smash{\scrs 0.8}}}%
    \put(0.50846974,0.00148765){\color[rgb]{0,0,0}\makebox(0,0)[b]{\smash{\scrs 1.0}}}%
    \put(0.59198183,0.00148765){\color[rgb]{0,0,0}\makebox(0,0)[b]{\smash{\scrs 1.2}}}%
    \put(0.67549393,0.00148765){\color[rgb]{0,0,0}\makebox(0,0)[b]{\smash{\scrs 1.4}}}%
    \put(0.75900607,0.00148765){\color[rgb]{0,0,0}\makebox(0,0)[b]{\smash{\scrs 1.6}}}%
    \put(0.84251822,0.00148765){\color[rgb]{0,0,0}\makebox(0,0)[b]{\smash{\scrs 1.8}}}%
    \put(0.07346775,0.11568219){\color[rgb]{0,0,0}\makebox(0,0)[rb]{\smash{\scrs 0.2}}}%
    \put(0.07346775,0.19919424){\color[rgb]{0,0,0}\makebox(0,0)[rb]{\smash{\scrs 0.4}}}%
    \put(0.07346775,0.28270638){\color[rgb]{0,0,0}\makebox(0,0)[rb]{\smash{\scrs 0.6}}}%
    \put(0.07346775,0.36621892){\color[rgb]{0,0,0}\makebox(0,0)[rb]{\smash{\scrs 0.8}}}%
    \put(0.07346775,0.44973102){\color[rgb]{0,0,0}\makebox(0,0)[rb]{\smash{\scrs 1.0}}}%
  \end{picture}%
\endgroup

%% file: figs/f_dcCapa2_constituentRegion.pdf_tex

\begingroup
  \makeatletter
  \providecommand\color[2][]{%
    \errmessage{(Inkscape) Color is used for the text in Inkscape, but the package 'color.sty' is not loaded}
    \renewcommand\color[2][]{}%
  }
  \providecommand\transparent[1]{%
    \errmessage{(Inkscape) Transparency is used (non-zero) for the text in Inkscape, but the package 'transparent.sty' is not loaded}
    \renewcommand\transparent[1]{}%
  }
  \providecommand\rotatebox[2]{#2}
  \ifx\svgwidth\undefined
    \setlength{\unitlength}{249.44882813pt}
  \else
    \setlength{\unitlength}{\svgwidth}
  \fi
  \global\let\svgwidth\undefined
  \makeatother
  \begin{picture}(1,0.41025849)%
    \put(0,0){\includegraphics[width=\unitlength]{f_dcCapa2_constituentRegion.pdf}}%
    \put(0.77272728,0.03525849){\color[rgb]{0,0,0}\makebox(0,0)[lb]{\smash{$\Rdk 1$}}}%
    \put(0.22159091,0.13184942){\color[rgb]{0,0,0}\makebox(0,0)[rb]{\smash{$\Rdk 2$}}}%
    \put(0.36931757,0.37616758){\color[rgb]{0,0,0}\makebox(0,0)[rb]{\smash{$R$}}}%
    \put(0.54545453,0.31366756){\color[rgb]{0,0,0}\makebox(0,0)[b]{\smash{\fn \eqref{eq:dcCapa2_cond1}}}}%
    \put(0.40909091,0.01366807){\color[rgb]{0,0,0}\makebox(0,0)[b]{\smash{\fn \eqref{eq:dcCapa2_cond2}}}}%
    \put(0.34090909,0.21707668){\color[rgb]{0,0,0}\makebox(0,0)[b]{\smash{\fn \eqref{eq:dcCapa2_cond3}}}}%
    \put(0.62500002,0.20571305){\color[rgb]{0,0,0}\makebox(0,0)[b]{\smash{\fn \eqref{eq:dcCapa2_cond4}}}}%
    \put(0.43181816,0.1034403){\color[rgb]{0,0,0}\makebox(0,0)[b]{\smash{\fn \eqref{eq:dcCapa2_cond5}}}}%
    \put(0.47727272,0.21707668){\color[rgb]{0,0,0}\makebox(0,0)[b]{\smash{\fn \eqref{eq:dcCapa2_cond6}}}}%
  \end{picture}%
\endgroup

%% file: figs/f_exampleDisturbanceChannel2.pdf_tex

\begingroup
  \makeatletter
  \providecommand\color[2][]{%
    \errmessage{(Inkscape) Color is used for the text in Inkscape, but the package 'color.sty' is not loaded}
    \renewcommand\color[2][]{}%
  }
  \providecommand\transparent[1]{%
    \errmessage{(Inkscape) Transparency is used (non-zero) for the text in Inkscape, but the package 'transparent.sty' is not loaded}
    \renewcommand\transparent[1]{}%
  }
  \providecommand\rotatebox[2]{#2}
  \ifx\svgwidth\undefined
    \setlength{\unitlength}{249.44882813pt}
  \else
    \setlength{\unitlength}{\svgwidth}
  \fi
  \global\let\svgwidth\undefined
  \makeatother
  \begin{picture}(1,0.34352118)%
    \put(0,0){\includegraphics[width=\unitlength]{f_exampleDisturbanceChannel2.pdf}}%
    \put(0.1931818,0.00148765){\color[rgb]{0,0,0}\makebox(0,0)[rb]{\smash{$M \to X$}}}%
    \put(0.22727271,0.02421491){\color[rgb]{0,0,0}\makebox(0,0)[lb]{\smash{\fn $\{0,1,2,3\}$}}}%
    \put(0.48863634,0.30603331){\color[rgb]{0,0,0}\makebox(0,0)[lb]{\smash{\fn $0$}}}%
    \put(0.48863634,0.27762424){\color[rgb]{0,0,0}\makebox(0,0)[lb]{\smash{\fn $1$}}}%
    \put(0.48863634,0.24921512){\color[rgb]{0,0,0}\makebox(0,0)[lb]{\smash{\fn $2$}}}%
    \put(0.5795454,0.26398871){\color[rgb]{0,0,0}\makebox(0,0)[rb]{\smash{\fn $1$}}}%
    \put(0.48863634,0.2208061){\color[rgb]{0,0,0}\makebox(0,0)[lb]{\smash{\fn $3$}}}%
    \put(0.5795454,0.22080547){\color[rgb]{0,0,0}\makebox(0,0)[rb]{\smash{\fn $2$}}}%
    \put(0.5795454,0.30603331){\color[rgb]{0,0,0}\makebox(0,0)[rb]{\smash{\fn $0$}}}%
    \put(0.77272719,0.26285128){\color[rgb]{0,0,0}\makebox(0,0)[lb]{\smash{$Z_2$}}}%
    \put(0.61363636,0.28557859){\color[rgb]{0,0,0}\makebox(0,0)[lb]{\smash{\fn $\{0,1,2\}$}}}%
    \put(0.48863634,0.1469424){\color[rgb]{0,0,0}\makebox(0,0)[lb]{\smash{\fn $0$}}}%
    \put(0.48863634,0.11853333){\color[rgb]{0,0,0}\makebox(0,0)[lb]{\smash{\fn $1$}}}%
    \put(0.48863634,0.09012422){\color[rgb]{0,0,0}\makebox(0,0)[lb]{\smash{\fn $2$}}}%
    \put(0.5795454,0.13330628){\color[rgb]{0,0,0}\makebox(0,0)[rb]{\smash{\fn $0$}}}%
    \put(0.48863634,0.0617152){\color[rgb]{0,0,0}\makebox(0,0)[lb]{\smash{\fn $3$}}}%
    \put(0.5795454,0.09012427){\color[rgb]{0,0,0}\makebox(0,0)[rb]{\smash{\fn $1$}}}%
    \put(0.5795454,0.06171456){\color[rgb]{0,0,0}\makebox(0,0)[rb]{\smash{\fn $2$}}}%
    \put(0.77272719,0.10376035){\color[rgb]{0,0,0}\makebox(0,0)[lb]{\smash{$Z_1$}}}%
    \put(0.61363636,0.12648765){\color[rgb]{0,0,0}\makebox(0,0)[lb]{\smash{\fn $\{0,1,2\}$}}}%
    \put(0.80681817,0.00148765){\color[rgb]{0,0,0}\makebox(0,0)[lb]{\smash{$Y \to \hat M$}}}%
  \end{picture}%
\endgroup

%% file: figs/f_exampleDisturbanceChannel2_region.pdf_tex

\begingroup
  \makeatletter
  \providecommand\color[2][]{%
    \errmessage{(Inkscape) Color is used for the text in Inkscape, but the package 'color.sty' is not loaded}
    \renewcommand\color[2][]{}%
  }
  \providecommand\transparent[1]{%
    \errmessage{(Inkscape) Transparency is used (non-zero) for the text in Inkscape, but the package 'transparent.sty' is not loaded}
    \renewcommand\transparent[1]{}%
  }
  \providecommand\rotatebox[2]{#2}
  \ifx\svgwidth\undefined
    \setlength{\unitlength}{249.44882813pt}
  \else
    \setlength{\unitlength}{\svgwidth}
  \fi
  \global\let\svgwidth\undefined
  \makeatother
  \begin{picture}(1,0.52076706)%
    \put(0,0){\includegraphics[width=\unitlength]{f_exampleDisturbanceChannel2_region.pdf}}%
    \put(0.79554088,0.0605398){\color[rgb]{0,0,0}\makebox(0,0)[lb]{\smash{$\Rdk 1$}}}%
    \put(0.20463181,0.15713068){\color[rgb]{0,0,0}\makebox(0,0)[rb]{\smash{$\Rdk 2$}}}%
    \put(0.35804093,0.4980398){\color[rgb]{0,0,0}\makebox(0,0)[rb]{\smash{$R$}}}%
  \end{picture}%
\endgroup

%% file: figs/f_exampleDisturbanceChannel2_cut1.pdf_tex

\begingroup
  \makeatletter
  \providecommand\color[2][]{%
    \errmessage{(Inkscape) Color is used for the text in Inkscape, but the package 'color.sty' is not loaded}
    \renewcommand\color[2][]{}%
  }
  \providecommand\transparent[1]{%
    \errmessage{(Inkscape) Transparency is used (non-zero) for the text in Inkscape, but the package 'transparent.sty' is not loaded}
    \renewcommand\transparent[1]{}%
  }
  \providecommand\rotatebox[2]{#2}
  \ifx\svgwidth\undefined
    \setlength{\unitlength}{249.44882813pt}
  \else
    \setlength{\unitlength}{\svgwidth}
  \fi
  \global\let\svgwidth\undefined
  \makeatother
  \begin{picture}(1,0.5640103)%
    \put(0,0){\includegraphics[width=\unitlength]{f_exampleDisturbanceChannel2_cut1.pdf}}%
    \put(0.09090908,0.5412603){\color[rgb]{0,0,0}\makebox(0,0)[rb]{\smash{$\Rd$}}}%
    \put(0.09619498,0.13135043){\color[rgb]{0,0,0}\makebox(0,0)[rb]{\smash{\scrs 0.5}}}%
    \put(0.09619498,0.23053061){\color[rgb]{0,0,0}\makebox(0,0)[rb]{\smash{\scrs 1.0}}}%
    \put(0.096195,0.32971094){\color[rgb]{0,0,0}\makebox(0,0)[rb]{\smash{\scrs 1.5}}}%
    \put(0.09619498,0.42889171){\color[rgb]{0,0,0}\makebox(0,0)[rb]{\smash{\scrs 2.0}}}%
    \put(0.92045452,0.00148765){\color[rgb]{0,0,0}\makebox(0,0)[lb]{\smash{$R$}}}%
    \put(0.27893688,0.00148775){\color[rgb]{0,0,0}\makebox(0,0)[b]{\smash{\scrs 0.5}}}%
    \put(0.44423747,0.00148775){\color[rgb]{0,0,0}\makebox(0,0)[b]{\smash{\scrs 1.0}}}%
    \put(0.6095355,0.00148765){\color[rgb]{0,0,0}\makebox(0,0)[b]{\smash{\scrs 1.5}}}%
    \put(0.77483612,0.00148765){\color[rgb]{0,0,0}\makebox(0,0)[b]{\smash{\scrs 2.0}}}%
    \put(0.19318185,0.41626037){\color[rgb]{0,0,0}\makebox(0,0)[lb]{\smash{\scrs Symmetric disturbance constraints}}}%
    \put(0.19318185,0.45035125){\color[rgb]{0,0,0}\makebox(0,0)[lb]{\smash{\scrs Single disturbance constraint }}}%
  \end{picture}%
\endgroup

%% file: figs/f_exampleDisturbanceChannel2_cut2.pdf_tex

\begingroup
  \makeatletter
  \providecommand\color[2][]{%
    \errmessage{(Inkscape) Color is used for the text in Inkscape, but the package 'color.sty' is not loaded}
    \renewcommand\color[2][]{}%
  }
  \providecommand\transparent[1]{%
    \errmessage{(Inkscape) Transparency is used (non-zero) for the text in Inkscape, but the package 'transparent.sty' is not loaded}
    \renewcommand\transparent[1]{}%
  }
  \providecommand\rotatebox[2]{#2}
  \ifx\svgwidth\undefined
    \setlength{\unitlength}{249.44882813pt}
  \else
    \setlength{\unitlength}{\svgwidth}
  \fi
  \global\let\svgwidth\undefined
  \makeatother
  \begin{picture}(1,0.58671496)%
    \put(0,0){\includegraphics[width=\unitlength]{f_exampleDisturbanceChannel2_cut2.pdf}}%
    \put(0.80681816,0.00148765){\color[rgb]{0,0,0}\makebox(0,0)[lb]{\smash{$\Rdk 1$}}}%
    \put(0.21590909,0.55262408){\color[rgb]{0,0,0}\makebox(0,0)[rb]{\smash{$\Rdk 2$}}}%
    \put(0.36213661,0.00148785){\color[rgb]{0,0,0}\makebox(0,0)[b]{\smash{\scrs 0.5}}}%
    \put(0.48563686,0.00148785){\color[rgb]{0,0,0}\makebox(0,0)[b]{\smash{\scrs 1.0}}}%
    \put(0.60913712,0.00148775){\color[rgb]{0,0,0}\makebox(0,0)[b]{\smash{\scrs 1.5}}}%
    \put(0.22119494,0.15567045){\color[rgb]{0,0,0}\makebox(0,0)[rb]{\smash{\scrs 0.5}}}%
    \put(0.22119494,0.27917056){\color[rgb]{0,0,0}\makebox(0,0)[rb]{\smash{\scrs 1.0}}}%
    \put(0.22119499,0.40267087){\color[rgb]{0,0,0}\makebox(0,0)[rb]{\smash{\scrs 1.5}}}%
    \put(0.77840904,0.40267126){\color[rgb]{0,0,0}\makebox(0,0)[rb]{\smash{\fn $R\!=\!2.0$}}}%
    \put(0.77839201,0.33291312){\color[rgb]{0,0,0}\makebox(0,0)[rb]{\smash{\fn $1.9$}}}%
    \put(0.77839221,0.28991848){\color[rgb]{0,0,0}\makebox(0,0)[rb]{\smash{\fn $1.8$}}}%
    \put(0.7783925,0.25203481){\color[rgb]{0,0,0}\makebox(0,0)[rb]{\smash{\fn $1.7$}}}%
    \put(0.77839221,0.21655659){\color[rgb]{0,0,0}\makebox(0,0)[rb]{\smash{\fn $1.6$}}}%
    \put(0.77839221,0.18308264){\color[rgb]{0,0,0}\makebox(0,0)[rb]{\smash{\fn $1.5$}}}%
    \put(0.7783925,0.15101184){\color[rgb]{0,0,0}\makebox(0,0)[rb]{\smash{\fn $1.4$}}}%
    \put(0.77839231,0.11954236){\color[rgb]{0,0,0}\makebox(0,0)[rb]{\smash{\fn $1.3$}}}%
    \put(0.7783925,0.09017772){\color[rgb]{0,0,0}\makebox(0,0)[rb]{\smash{\fn $1.2$}}}%
    \put(0.7783925,0.06121357){\color[rgb]{0,0,0}\makebox(0,0)[rb]{\smash{\fn $1.1$}}}%
  \end{picture}%
\endgroup

%% file: figs/f_dcCapa_constituentRegion_Gaussian.pdf_tex

\begingroup
  \makeatletter
  \providecommand\color[2][]{%
    \errmessage{(Inkscape) Color is used for the text in Inkscape, but the package 'color.sty' is not loaded}
    \renewcommand\color[2][]{}%
  }
  \providecommand\transparent[1]{%
    \errmessage{(Inkscape) Transparency is used (non-zero) for the text in Inkscape, but the package 'transparent.sty' is not loaded}
    \renewcommand\transparent[1]{}%
  }
  \providecommand\rotatebox[2]{#2}
  \ifx\svgwidth\undefined
    \setlength{\unitlength}{249.44882813pt}
  \else
    \setlength{\unitlength}{\svgwidth}
  \fi
  \global\let\svgwidth\undefined
  \makeatother
  \begin{picture}(1,0.57309708)%
    \put(0,0){\includegraphics[width=\unitlength]{f_dcCapa_constituentRegion_Gaussian.pdf}}%
    \put(0.62784881,0.41332727){\color[rgb]{0,0,0}\makebox(0,0)[lb]{\smash{$A$}}}%
    \put(0.46875793,0.24287276){\color[rgb]{0,0,0}\makebox(0,0)[lb]{\smash{$B$}}}%
    \put(0.54830337,0.32241816){\color[rgb]{0,0,0}\makebox(0,0)[lb]{\smash{\small 45\degree}}}%
    \put(0.684667,0.08378178){\color[rgb]{0,0,0}\makebox(0,0)[lb]{\smash{$R$}}}%
    \put(0.32103067,0.56105453){\color[rgb]{0,0,0}\makebox(0,0)[rb]{\smash{$\Rd$}}}%
    \put(0.46307611,0.40196365){\color[rgb]{0,0,0}\makebox(0,0)[b]{\smash{$\Rr(U,X)$}}}%
    \put(0.57103067,0.03264541){\color[rgb]{0,0,0}\makebox(0,0)[lb]{\smash{\small $\tfrac 1 2 \log \frac
{\det{K_u+K_v+K_1}}
{\det{K_1}}$}}}%
    \put(0.52557611,0.03264541){\color[rgb]{0,0,0}\makebox(0,0)[rb]{\smash{\small $\tfrac 1 2 \log \frac
{\det{K_v+K_1}}
{\det{K_1}}$}}}%
    \put(0.30966705,0.25423639){\color[rgb]{0,0,0}\makebox(0,0)[rb]{\smash{\small $\tfrac 1 2 \log \frac
{\det{K_v+K_2}}
{\det{K_2}}$}}}%
  \end{picture}%
\endgroup

%% file: figs/f_dcCodingScheme_marton_monochrome.pdf_tex

\begingroup
  \makeatletter
  \providecommand\color[2][]{%
    \errmessage{(Inkscape) Color is used for the text in Inkscape, but the package 'color.sty' is not loaded}
    \renewcommand\color[2][]{}%
  }
  \providecommand\transparent[1]{%
    \errmessage{(Inkscape) Transparency is used (non-zero) for the text in Inkscape, but the package 'transparent.sty' is not loaded}
    \renewcommand\transparent[1]{}%
  }
  \providecommand\rotatebox[2]{#2}
  \ifx\svgwidth\undefined
    \setlength{\unitlength}{249.44882813pt}
  \else
    \setlength{\unitlength}{\svgwidth}
  \fi
  \global\let\svgwidth\undefined
  \makeatother
  \begin{picture}(1,0.67505571)%
    \put(0,0){\includegraphics[width=\unitlength]{f_dcCodingScheme_marton_monochrome.pdf}}%
    \put(0.38636363,0.6628726){\color[rgb]{0,0,0}\makebox(0,0)[b]{\smash{$z_2^n(1,l_2^{(1,1)})$}}}%
    \put(0.4375,0.00378178){\color[rgb]{0,0,0}\makebox(0,0)[b]{\smash{$m_2=1$}}}%
    \put(0.65340907,0.00378178){\color[rgb]{0,0,0}\makebox(0,0)[b]{\smash{$m_2=2$}}}%
    \put(0.86931819,0.00378178){\color[rgb]{0,0,0}\makebox(0,0)[b]{\smash{$m_2=3$}}}%
    \put(0.96590907,0.48673628){\color[rgb]{0,0,0}\rotatebox{90}{\makebox(0,0)[b]{\smash{$m_1=1$}}}}%
    \put(0.96590907,0.27082717){\color[rgb]{0,0,0}\rotatebox{90}{\makebox(0,0)[b]{\smash{$m_1=2$}}}}%
    \put(0.34659088,0.40719079){\color[rgb]{0,0,0}\makebox(0,0)[lb]{\smash{\small $x^n$\fn$(1,1,1,1)$}}}%
    \put(0.19318181,0.20832717){\color[rgb]{0,0,0}\makebox(0,0)[b]{\smash{$y^n$}}}%
    \put(0.28409091,0.57196354){\color[rgb]{0,0,0}\makebox(0,0)[rb]{\smash{\fn $z_1^n(1,1)$}}}%
    \put(0.28409091,0.53787261){\color[rgb]{0,0,0}\makebox(0,0)[rb]{\smash{\fn $z_1^n(1,2)$}}}%
    \put(0.26136365,0.48673628){\color[rgb]{0,0,0}\makebox(0,0)[rb]{\smash{$z_1^n(1,l_1^{(1,1)})$}}}%
    \put(0.28409091,0.3901453){\color[rgb]{0,0,0}\makebox(0,0)[rb]{\smash{\fn $z_1^n(1,2^{n(\tilde R_1-R_1)})$}}}%
  \end{picture}%
\endgroup

%% file: figs/f_dcCapa2_constituentRegion_roof.pdf_tex

\begingroup
  \makeatletter
  \providecommand\color[2][]{%
    \errmessage{(Inkscape) Color is used for the text in Inkscape, but the package 'color.sty' is not loaded}
    \renewcommand\color[2][]{}%
  }
  \providecommand\transparent[1]{%
    \errmessage{(Inkscape) Transparency is used (non-zero) for the text in Inkscape, but the package 'transparent.sty' is not loaded}
    \renewcommand\transparent[1]{}%
  }
  \providecommand\rotatebox[2]{#2}
  \ifx\svgwidth\undefined
    \setlength{\unitlength}{249.44882813pt}
  \else
    \setlength{\unitlength}{\svgwidth}
  \fi
  \global\let\svgwidth\undefined
  \makeatother
  \begin{picture}(1,0.42505363)%
    \put(0,0){\includegraphics[width=\unitlength]{f_dcCapa2_constituentRegion_roof.pdf}}%
    \put(0.77272732,0.03525847){\color[rgb]{0,0,0}\makebox(0,0)[lb]{\smash{$\Rdk 1$}}}%
    \put(0.22159089,0.1318494){\color[rgb]{0,0,0}\makebox(0,0)[rb]{\smash{$\Rdk 2$}}}%
    \put(0.36931815,0.41025849){\color[rgb]{0,0,0}\makebox(0,0)[rb]{\smash{$R$}}}%
    \put(0.54545452,0.33071303){\color[rgb]{0,0,0}\makebox(0,0)[b]{\smash{\fn \eqref{eq:dcCapa2_roof1}}}}%
    \put(0.41477271,0.13753121){\color[rgb]{0,0,0}\makebox(0,0)[b]{\smash{\fn \eqref{eq:dcCapa2_roof2}}}}%
    \put(0.64772722,0.0920767){\color[rgb]{0,0,0}\makebox(0,0)[b]{\smash{\fn \eqref{eq:dcCapa2_roof3}}}}%
    \put(0.53409089,0.08071303){\color[rgb]{0,0,0}\makebox(0,0)[b]{\smash{\fn \eqref{eq:dcCapa2_roof4}}}}%
  \end{picture}%
\endgroup